\documentclass[11pt]{article}
\usepackage{amsmath, afterpage, pdflscape, changepage, multirow, multicol, cleveref, hhline, array, xcolor, verbatim, graphicx, url, fancyhdr, etoolbox, bigstrut, cite, setspace, float, xr, amssymb, mathtools, capt-of, amsthm, bm, hhline, blkarray, pbox, pdfpages, bbm, extarrows}
\usepackage[font=small,labelfont=bf]{caption}
\usepackage[hang]{footmisc}
\usepackage[USenglish]{babel}
\usepackage[USenglish]{isodate}
\usepackage[hmargin=1.5cm,top=1.5cm,bottom=1.5cm]{geometry}
\usepackage[inline]{enumitem}

\AtBeginEnvironment{tabular}{\onehalfspacing}

\Crefname{equation}{Eq}{Eqs}

\newcommand{\C}{\mathcal{C}}

\newcommand{\RR}{\mathbb{R}}
\newcommand{\ZZ}{\mathbb{Z}}

\let\svthefootnote\thefootnote

\newcommand\blankfootnote[1]{%
  \let\thefootnote\relax\footnotetext{#1}%
  \let\thefootnote\svthefootnote%
}

\newtheorem{theorem}{Theorem}[section]
\newtheorem{lemma}[theorem]{Lemma}
\newtheorem{proposition}[theorem]{Proposition}

\theoremstyle{remark}
\newtheorem{definition}[theorem]{Definition}

\newtheorem{example}[theorem]{Example}

\newtheorem{remark}[theorem]{Remark}

\newtheorem{notation}[theorem]{Notation}

\def\D{{\mathcal D}}

\def\RR{{\mathbb R}}

\def\C{{\mathcal C}}

\def\G{{\mathcal G}}
\def\K{{\mathcal K}}

\def\O{{\mathcal O}}

\def\SC{{\mathcal S}}

\def\X{{\mathcal X}}
\def\ZZ{{\mathbb Z}}

\newcommand{\Aut}{\operatorname{Aut}}

\newcommand{\id}{\operatorname{id}}
\newcommand{\Ord}{\operatorname{Ord}}

\newcommand{\Vertices}{\operatorname{Vert}}

\newcommand{\Sym}{\operatorname{Sym}}

\makeatletter
\def\namedlabel#1#2{\begingroup
   \def\@currentlabel{#2}%
   \label{#1}\endgroup
}
\makeatother

\title{\LARGE A group theoretic approach to model comparison with\\ simplicial representations}

\author{\normalsize Sean T. Vittadello\footnote{School of Mathematics and Statistics and School of BioSciences, The University of Melbourne, Parkville, Victoria 3010, Australia\label{VS}}\textsuperscript{,}\renewcommand*{\thefootnote}{\fnsymbol{footnote}}\footnote[1]{Corresponding author: sean.vittadello@unimelb.edu.au}\renewcommand*{\thefootnote}{\arabic{footnote}} and Michael P.H. Stumpf\textsuperscript{\ref{VS}}}

\fancypagestyle{specialfooter}{%
\fancyhf{}

\fancyfoot[R]{\footnotesize \emph{\today}}
}

\begin{document}
\captionsetup[figure]{name={Figure}}
\makeatletter
\begin{titlepage}
\thispagestyle{specialfooter}
\blankfootnote{\textbf{Key words and phrases}: Model comparison, model similarity, model equivalence, simplicial complex, group action, orbit space}
\centering
\@title\\
\vspace{1.5cm}
\@author\\
\vspace{1.5cm}
\cleanlookdateon
\vspace{0mm}
\begin{abstract}
\begin{normalsize}\noindent
The complexity of biological systems, and the increasingly large amount of associated experimental data, necessitates that we develop mathematical models to further our understanding of these systems. Because biological systems are generally not well understood, most mathematical models of these systems are based on experimental data, resulting in a seemingly heterogeneous collection of models that ostensibly represent the same system. To understand the system we therefore need to understand how the different models are related to each other, with a view to obtaining a unified mathematical description. This goal is complicated by the fact that a number of distinct mathematical formalisms may be employed to represent the same system, making direct comparison of the models very difficult. A methodology for comparing mathematical models based on their underlying conceptual structure is therefore required. In previous work we developed an appropriate framework for model comparison where we represent models, specifically the conceptual structure of the models, as labelled simplicial complexes and compare them with the two general methodologies of comparison by distance and comparison by equivalence. In this article we continue the development of our model comparison methodology in two directions. First, we present a rigorous and automatable methodology for the core process of comparison by equivalence, namely determining the vertices in a simplicial representation, corresponding to model components, that are conceptually related and the identification of these vertices via simplicial operations. Our methodology is based on considerations of vertex symmetry in the simplicial representation, for which we develop the required mathematical theory of group actions on simplicial complexes. This methodology greatly simplifies and expedites the process of determining model equivalence. Second, we provide an alternative mathematical framework for our model-comparison methodology by representing models as groups, which allows for the direct application of group-theoretic techniques within our model-comparison methodology.
\end{normalsize}
\end{abstract}
\end{titlepage}
\makeatother

\newpage

\section{Introduction}
Describing biological systems mathematically remains a challenge. We can collect data at unprecedented scales \cite{Howe2008,Marx2013,ISB2018,Mahmud2021}, but we still lack the models required to integrate the often heterogeneous data into a single analytical and predictive framework. Mathematical models can achieve this, but also much more. They are of fundamental importance in developing new knowledge of these complex systems for multiple reasons: analysis and interpretation of large data sets; simulations of the biological system, or a subsystem thereof, to develop new hypotheses; and, guiding new experimental design \cite{Tomlin2007,Sneddon2010,Gunawardena2014,Wolkenhauer2014,Torres2015,Pezzulo2016,Transtrum2016,Banwarth_Kuhn2020,King2021}.

Formally, a \emph{model} is an abstraction of an observable phenomenon \cite{Rosenblueth1945}. Models may assume many different forms, and here we consider two general classes, conceptual models and mathematical models \cite{Torres2015}. A \emph{conceptual model} is a qualitative representation of the reference system under investigation, consisting of the relevant concepts which correspond to observable objects or phenomena, and the interconnections between these concepts. A conceptual model is obtained by integrating our knowledge of a system into a descriptive framework, with a range of possible levels of conceptual detail depending on our knowledge of the system. For example, in a developmental-patterning system two possible concepts are a \emph{morphogen} and \emph{diffusion of the morphogen}, which are dependent concepts in the system since the morphogen undergoes diffusion. Some developmental-patterning systems may have a second morphogen, which is unrelated to the first morphogen since they are independent concepts. A \emph{mathematical model} is a quantitative representation of a conceptual model, and hence a representation of the reference system, in the formalism of mathematical concepts. Mathematical models are based on conceptual models through specification of the relevant mathematical concepts and their relationships. Mathematical models allow for the calculation of quantifiable concepts of the system, for example the spatial distribution of the morphogen concentration as a function of time.

Comparison of mathematical models based on their underlying conceptual structure is of fundamental importance for understanding the corresponding biological systems. Many aspects of complex biological systems are not well understood, so the development of mathematical models of such systems necessarily follows an inductive process employing experimental data which provides limited constraint for the range of possible conceptual models, resulting in non-uniqueness of the mathematical models \cite{Karplus1977,Babtie2014}. While mathematical models make conceptual models more specific and precise, such specificity may provide room for many alternative mathematical representations. Indeed, there are many possible formulations of a quantitative mathematical model from a qualitative conceptual model, depending on the particular quantitative description required, and alternative mathematical models can appear very distinct and unrelated since the underlying concepts may be obfuscated by the mathematical formalism or may be interpretable in different ways. Further, while the mathematical model is intended to be in direct correspondence with the conceptual model, and therefore the reference system, the particular formalism of the mathematical model introduces many implicit connections between the mathematical concepts in the mathematical model. Some of these interconnections may not be intended, and may be inconsistent with the underlying conceptual model; therefore the mathematical model may not correspond directly to the conceptual model. It is therefore essential to consider carefully whether the underlying conceptual model of the mathematical model represents the reference system. Indeed, since the mathematical formalism may introduce conceptual interconnections unrelated to the original system, care must be taken to ensure that the model is not a weak or incorrect representation of the system of interest.

Because of the possible multiplicity, different mathematical formalisms, and misrepresentation of mathematical models associated with a biological system, a general methodology for model comparison is required in order to further understanding of the system underlying the different types of models \cite{Gay2010,Henkel2018}. Model comparison is therefore of fundamental importance in understanding the similarities and differences between models and between data sets, and consequently in understanding biological systems. While there are many approaches to compare the quantitative outputs of mathematical models, see for example \cite{Tapinos2013} and \cite{Cabbia2020}, such comparisons provide little insight into the underlying conceptual structure of the mathematical models and hence of the corresponding biological systems. Neither can we rely on model selection \cite{Kirk2013} or robustness analysis, as this also relies on a comparison of the outputs of different models. These approaches cannot be used \emph{a priori} to determine conceptual similarities, equivalences, or, indeed, differences between mathematical models.

Model comparison can simplify the landscape of mathematical modelling by organising models into categories of models that are closely related. For example, models in the same category may be incremental variations of each other, but share some fundamental qualitative characteristic. We are interested in this fundamental level of agreement: the conceptual basis of a model. Or alternatively, the set of design principles \cite{Barnes2011,Brophy2014} shared by all models in the same category.  The mathematical models are merely tools for quantitative computations, and the mathematical formalism can in fact be a hindrance to the model comparison at the conceptual level.

In Vittadello and Stumpf \cite{Vittadello2021b} we present a framework for model comparison where models are represented as labelled simplicial complexes. Models are then compared in terms of their simplicial representations using two general methodologies: the first is comparison by distance, where the difference between models is measured in terms of the differences between the corresponding simplicial representations; the second is comparison by equivalence, where we identify any equivalences between the components of the simplicial representations of different models, and then employ particular operations on the complexes to transform one into the other --- such equivalence, when it exists, reveals common underlying characteristics of the models. Our methodology for model comparison allows for any model to be represented and compared within our framework, irrespective of formalism, granularity, and level of abstraction. To construct the simplicial representation of a model, we must first define the components of the model that we would like to represent, and for this we can fall back on relevant domain expertise. Depending on the level of detail required, the simplicial representation of a model may not be unique. One main outcome of our model-comparison methodology is for models of developmental-pattern formation \cite{Green2015,Scholes2019}, where we demonstrate that the Turing-pattern activator-inhibitor model is equivalent to the positional-information annihilation model from a significant conceptual perspective.

In this article we continue the development of our methodology for model comparison \cite{Vittadello2021b} in two respects. First, we develop a rigorous and automatable methodology for the main process when comparing models with respect to equivalence: the determination of the vertices that are conceptually related in a simplicial representation and the subsequent identification of these vertices via simplicial operations. We say that vertices are conceptually related when their labels, corresponding to model components, are conceptually related. Without automatability, establishing which vertices are conceptually related in a simplicial representation is challenging if not unfeasible, even for relatively small models. This methodology is based on considerations of vertex symmetry in the simplicial representation, and we establish the necessary mathematical theory for group actions on simplicial complexes. The methodology provides a major step toward simplifying and facilitating the process of determining model equivalence, which in turn can assist in the unification of models for a biological system and the discovery of design principles for synthetic biology. Second, we present a concise description of an alternative mathematical framework for our model-comparison methodology where we represent models as groups. This allows for the direct application of group-theoretic techniques within our model-comparison methodology, in addition to the current algebraic-topological techniques associated with simplicial complexes.

The remainder of this article is organised as follows. In Section~\ref{sec:Background} we provide an overview of the required background material: Subsection~\ref{subsec:Topology} outlines the required background in algebraic topology; we then describe our methodology for model comparison by equivalence in Subsection~\ref{subsec:ModelEquiv}. We present and discuss our new results in Section~\ref{sec:Discussion}. In particular, in Subsection~\ref{subsec:Overview} we provide a descriptive overview of our approach for identifying equivalent components in simplicial representations through group actions, and then develop this mathematical theory in Subsection~\ref{subsec:Gaction}; we present a step-by-step methodology for determining model equivalence in Subsection~\ref{subsec:method}, which employs our automatable theory from Subsection~\ref{subsec:Gaction}; we then apply our methodology to a collection of Turing-pattern and positional-information models (see the Appendix in Section~\ref{sec:Appendix} for the model details) in Subsection~\ref{subsec:example}; in Subsection~\ref{subsec:Grep} we give an incisive description of an alternative mathematical framework for our model-comparison methodology by representing models as groups, which we call $G$-representations. Finally, in Section~\ref{sec:Conclusion}, we conclude with a discussion of the utility of our new theory for comparing models.

\section{Background} \label{sec:Background}
In this section we provide a brief and informal description of the necessary background in simplicial algebraic topology; further details are available from standard references \cite{Spanier1966,Rotman1988,Munkres2018}. We also discuss our methodology for model comparison by equivalence as described in Vittadello and Stumpf \cite{Vittadello2021b}, and developed further in this article.

\subsection{Algebraic topology}\label{subsec:Topology}
Geometrically a \emph{$p$-simplex} is a generalisation of a filled triangle to an arbitrary dimension $p$, whereby a point is a 0-simplex, a line segment is a 1-simplex, a filled triangle is a 2-simplex, a filled tetrahedron is a 3-simplex, and so forth. We may regard 0-simplices as vertices and 1-simplices as edges, so that we consider a simple graph as a collection of 0-simplices and 1-simplices. A \emph{$k$-face} of a simplex is a $k$-dimensional subsimplex, and the 0-faces (or vertices) of a simplex \emph{span} the simplex. If the simplex $\tau$ is a face of the simplex $\sigma$ then we write $\tau \le \sigma$. A face $\tau$ of a simplex $\sigma$ is \emph{proper} if $\tau \ne \sigma$.

A \emph{simplicial complex} is a generalisation of a simple graph, whereby simplicies of dimension higher than one represent higher-dimensional interactions between the vertices. Specifically, a simplicial complex $K$ consists of a set of simplices such that: if $\sigma$ is a simplex in $K$ then every face of $\sigma$ is also in $K$; and, the nonempty intersection of any two simplices in $K$ is a simplex in $K$. We denote the set of vertices of $K$ as $\Vertices (K)$. A \emph{simplicial subcomplex} $L$ of a simplicial complex $K$ is a collection $L \subseteq K$ that is also a simplicial complex. A simplicial subcomplex $L$ of a simplicial complex $K$ is a \emph{full subcomplex} if every simplex of $K$ that has all its vertices in $L$ is also a simplex of $L$. The \emph{$k$-skeleton} of a simplicial complex $K$ is the subcomplex $K^{(k)}$ consisting of the simplices in $K$ with dimension at most $k$. In particular, the 0-skeleton $K^{(0)}$ is the set of vertices and the 1-skeleton $K^{(1)}$ is the \emph{underlying graph} of the simplicial complex $K$.

In some contexts it is useful to define the \emph{empty simplex}, denoted $\emptyset$, as a $(-1)$-simplex, and then the corresponding simplicial complex is the \emph{empty simplicial complex}, denoted $\{ \emptyset\}$, which has dimension $-1$. In this work we exclude the empty simplex, as it is not required. We include, however, the \emph{void simplicial complex}, which is the simplicial complex consisting of the empty collection of sets. We denote the void simplicial complex by $\emptyset$, which has the assigned dimension of $-\infty$. Further details on these degenerate cases are found in \cite[Page 8, Remark 2.3]{Kozlov2008}.

We shall move freely between the interpretations of simplicial complexes as both geometric and abstract objects, since they are equivalent and both contain the same combinatorial information; and we have no need for the topological aspects of geometric simplicial complexes. Notationally, if $v_1, v_2, \ldots, v_n$ are vertices that span a simplex $K$ then we denote $K$ as $\{v_1, v_2, \ldots, v_n \}$ when we want to emphasise the spanning vertices. Note that vertices are therefore denoted as both $v_i$ and $\{v_i\}$.

\subsection{Model comparison by equivalence} \label{subsec:ModelEquiv}
We refer to a specific conceptual detail of a given model as a \emph{component} of the model. Every model consists of a finite number of components and the conceptual interconnections between these components. For example, a Turing-pattern system of reaction-diffusion equations is a model of a developmental-patterning process, and includes components such as morphogens, diffusion, reactions, boundary conditions, and a morphogen gradient, along with the interconnections between particular components such as a morphogen and its diffusion.

To compare a collection of models we first determine the set of all model components, at the required level of conceptual detail, associated with the models in the collection. The model components are the scientific concepts represented within the model, interpreted with knowledge of the relevant scientific domains.

\begin{definition}[\bf{Model components, generated model}]\label{def:Model components}
Let $\C$ be the set of all required components from the collection of models under consideration for comparison. We say that each model in the collection is \emph{generated} by a subset of components from $\C$.
\end{definition}

We represent each model in the collection for comparison as a labelled simplicial complex, which we call a \emph{simplicial representation}, where each model component is represented as a labelled 0-simplex. The label associated with a 0-simplex corresponds to the relevant concept in the model, and for simplicity and ease of automation could have a numerical label:

\begin{notation}[\bf{Representations of model components}]\label{not:Model components}
Let $\Ord \colon \C \to \{1, 2, \ldots, |\C|\}$ be a bijection that specifies an arbitrary order for the categorical data elements in $\C$. We represent each model component from $\C$ as a 0-simplex that is labelled with the name of the model component or the corresponding positive integer under $\Ord$.
\end{notation}

\begin{definition}[\bf{Simplicial representation}]\label{def:SimpRep}
Let $\C$ be the set of all required components from the collection of models under consideration for comparison. To each model generated by a subset of components from $\C$ we associate a labelled simplicial complex, which we call a \emph{simplicial representation}, where the labelled 0-simplices correspond to components of the model, and the one- and higher-dimensional simplices correspond to the conceptual interconnections between the model components as determined by the model.
\end{definition}

The ordering function $\Ord$ provides for efficient labelling of simplices. For example, a 2-simplex can be labelled as $\{1, 2, 3\}$, where the 0-simplices are $\{1\}$, $\{2\}$, and $\{3\}$, and the 1-simplices are $\{1,2\}$, $\{1,3\}$, and $\{2,3\}$.

For a particular level of component detail of a given model we associate a simplicial representation consisting of labelled 0-simplices, which represent the model components, along with the one- and higher-dimensional simplices that represent the conceptual interconnections between the components: the dyadic interactions are 1-simplices, the triadic interactions are 2-simplices, and so forth as required. The labelling of the 0-simplices induces a labelling of the higher-dimensional simplices through their spanning 0-simplices. Since the models for comparison are generated by the same set of components $\C$, a particular labelled simplex representing specific components and interconnections can be identified within different simplicial representations associated with the models.

Our definition of the equivalence of two models is in terms of a formal equivalence between the two corresponding simplicial representations. The formal equivalence is based on partial operations on the relevant set of simplicial representations, which we make explicit here.

\begin{definition}[\textbf{Partial operation on a set}]
Let $\X$ be a set. A \emph{partial operation} on $\X$ is a partial unary function $f \colon \X \rightharpoonup \X$. Denoting the domain of definition of $f$ as $\D \subseteq \X$, we say that $f$ is \emph{invertible} if there exists a partial operation $g \colon \X \rightharpoonup \X$ such that $(g \circ f) |_{\D} = \id_{\D}$ and $(f \circ g) |_{f(\D)} = \id_{f(\D)}$.
\end{definition}

Note that the domain of definition of a partial operation may have cardinality one.

\begin{definition}[\textbf{Equivalence of simplicial representations}]\label{def:Equiv}
Let $\C$ be the set of all components that appear in the collection of models under consideration for comparison, and let $\SC$ be a collection of simplicial representations corresponding to models generated by subsets from $\C$. Further, let $\O$ be a set of partial operations on $\SC$. Two simplicial representations $K$, $L \in \SC$ are \emph{equivalent} with respect to $\O$ if and only if there exists a (possibly empty) sequence of invertible partial operations $( f_i )_{i=0}^n$ in $\O$ such that $f_n \circ \cdots \circ f_1 \circ f_0 (K) = L$.
\end{definition}

\begin{remark}
Note that the equivalence of simplicial representations is always relative to the set $\O$ of partial operations on $\SC$. The choice of the set $\O$ determines which model components will be considered as conceptually equivalent. We employ partial operations on $\SC$ that we call \emph{admissible}, meaning that they only relate model components that we consider to be conceptually equivalent. We discuss this further below when we define the operations that we employ.
\end{remark}

\begin{remark}
It follows from our notion of equivalence of simplicial representations in Definition~\ref{def:Equiv} that the set $\{\, (K,L) \in \SC \times \SC \mid \text{$K$ and $L$ are equivalent} \,\}$ is an equivalence relation on $\SC$.
\end{remark}

\begin{remark}\label{rem:AltEquiv}
Note that the equivalence of two simplicial representations $K$, $L \in \SC$ can also be defined as follows: $K$ and $L$ are \emph{equivalent} if and only if there exist two (possibly empty) sequences of invertible partial operations $( f_i )_{i=0}^m$ and $( g_i )_{i=0}^n$ on $\SC$ such that $f_m \circ \cdots \circ f_1 \circ f_0 (K) = g_n \circ \cdots \circ g_1 \circ g_0 (L)$. In this work we use the fact that this form of equivalence of simplicial representations implies the form in Definition~\ref{def:Equiv}.
\end{remark}

We previously defined five partial operations that we use to establish equivalence of simplicial representations, namely: (Operation 1) Adjacent-vertex identification, (Operation 2) Nonadjacent-vertex identification, (Operation 3) Vertex split, (Operation 4) Inclusion, and (Operation 5) Vertex substitution. These five partial operations are each induced by a map between simplicial complexes, and are all invertible for suitable domains of definition: an adjacent-vertex identification is mutually inverse with a corresponding vertex split; a nonadjacent-vertex identification is mutually inverse with an inclusion; and, a vertex substitution is mutually inverse with another vertex substitution. From the perspective of Remark~\ref{rem:AltEquiv} we will only explicitly consider Operations 1, 2, and 5, while Operations 3 and 4 will be implicit.

For the equivalence of simplicial representations to be conceptually meaningful we need to ensure that the employed partial operations are themselves conceptually meaningful, or \emph{admissible}, meaning that the operations preserve the conceptual representation of the general physical system. Importantly, an established equivalence between models is based on the components of the models that we choose to identify as equivalent, and therefore the corresponding partial operations that we consider admissible. The equivalence of two models is therefore determined by the formal requirements in terms of partial operations, along with tight domain-specific constraints.

We now state Operations 1, 2, and 5. Further details are in Vittadello and Stumpf \cite{Vittadello2021b}. Recall that two distinct vertices in a simplicial complex are \emph{adjacent} if they belong to the same simplex, and for a vertex $u$ in a simplicial complex $K$ we denote the set of all vertices adjacent to $u$ as $V_{K} (u)$, noting that $u \notin V_{K} (u)$. Further, let $\C$ be the set of all components of models under consideration, and let $K$ and $L$ be two simplicial representations with labels from $\C$.

\begin{definition}[\textbf{Operation 1: Adjacent-vertex identification}]\label{def:Op1}
Let $u$ and $v$ be a pair of adjacent vertices in $K$ such that the following all hold:
\begin{itemize}[itemsep=5pt,topsep=5pt]
\item $V_{K} (u) \setminus \{v\} = V_{K} (v) \setminus \{u\}$.
\item For any nonempty subset $W \subseteq V_{K} (u) \setminus \{v\}$, the vertices $W \cup \{ u \}$ span a simplex in $K$ if and only if the vertices $W \cup \{ v \}$ span a simplex in $K$.
\end{itemize}
A simplicial map $\pi_1 \colon K \to L$ is an \emph{adjacent-vertex identification} if $\pi_1$ is surjective, and is injective and label preserving on every vertex except at the pair of vertices $u$ and $v$ which are mapped to a single vertex $c \in L^{(0)}$.
\end{definition}

\begin{definition}[\textbf{Operation 2: Nonadjacent-vertex identification}]\label{def:Op2}
Let $u$ and $v$ be a pair of nonadjacent vertices in $K$ such that the following all hold:
\begin{itemize}[itemsep=5pt,topsep=5pt]
\item $V_{K} (u)  = V_{K} (v) $.
\item For any nonempty subset $W \subseteq V_{K} (u)$, the vertices $W \cup \{ u \}$ span a simplex in $K$ if and only if the vertices $W \cup \{ v \}$ span a simplex in $K$.
\end{itemize}
A simplicial map $\pi_2 \colon K \to L$ is a \emph{nonadjacent-vertex identification} if $\pi_2$ is surjective, and is injective and label preserving on every vertex except at the pair of vertices $u$ and $v$ that are mapped to a single vertex $c \in L^{(0)}$.
\end{definition}

\begin{definition}[\textbf{Operation 5: Vertex substitution}]\label{def:Op5}
A simplicial map $\pi_5 \colon K \to L$ is a \emph{vertex substitution} if $\pi_5$ is bijective and preserves all labels except for one whereby the labelled vertex $u \in K^{(0)}$ is mapped to the labelled vertex $c \in L^{(0)}$.
\end{definition}

We may extend each partial operation to effect, for example, multiple vertex identifications or substitutions by extending the assumptions in Operations 1, 2, and 5.

\begin{remark}
We note here that our model-comparison methodology is designed to be flexible, so that the partial operations can be chosen based on specific requirements. Different operations may produce different model equivalences, however operations that correspond to very weak conceptual similarity may yield conceptual equivalences that are not scientifically meaningful.

The specific partial operations on simplicial representations that we employ (vertex identifications and substitution) are chosen to represent our notion of conceptual similarity for the models and underlying biological systems that we consider here. Of primary interest for our work are the vertex-identification operations, which identify vertices in a simplicial representation corresponding to conceptually-related model components; the identification of such vertices in a simplicial representation in a rigorous and automatable manner is the subject of Section~\ref{sec:Discussion}. Additional partial operations may be useful for grouping models based on equivalence, however the operations that we consider provide much flexibility while preserving conceptual equivalence.
\end{remark}

\section{Results and discussion}\label{sec:Discussion}
Here we have two main objectives. The first is to develop an automatable methodology for identifying equivalent components in models. The second is to provide a representation of models as groups, as an alternative to simplicial representations, with which models can be compared in a manner equivalent to that with simplicial representations. We begin with an overview of our approach for formalising the identification of equivalent model components with group actions.

\subsection{Overview}\label{subsec:Overview}
The most difficult aspect of establishing the equivalence of two simplicial representations is in determining whether two adjacent or two nonadjacent vertices in one simplicial representation can be identified conceptually with a vertex in the other simplicial representation, either via a vertex identification or vertex split operation. Identifiable vertices in a simplicial representation must be in positions that are conceptually equivalent, and in particular in symmetric positions within the corresponding unlabelled simplicial complex. We can therefore employ group actions on simplicial complexes to greatly simplify and expedite the process of determining the existence of identifiable vertices, while also providing for the process to be automated. The existence of a relevant nontrivial group action then induces a vertex-identification operation.

Given a simplicial representation, our goal is to determine whether any pairs of vertices are in symmetrical positions in the simplicial complex, and so are candidates for a possible vertex-identification operation. We therefore need to determine whether the simplicial complex has any symmetries of a certain class. By a symmetry of a simplicial complex we mean a permutation on the set of vertices of the complex that acts simplicially, so sends simplices to simplices, and that sends the whole complex to itself. We are interested in a class of symmetries, which we may call \emph{exchange symmetries}, of the complex that exchange two vertices and leave the other vertices fixed, since these symmetries reveal the vertices that are candidates for a possible vertex-identification operation. From a set of exchange symmetries of the complex we then obtain a new simplicial complex in which the exchangeable vertices are identified by a vertex-identification operation; this process can be iterated until a simplicial complex is obtained for which no exchange symmetries exist, and only relabelling of the vertices need be considered by a vertex-substitution operation.

\subsection{Identification of equivalent model components through group actions}\label{subsec:Gaction}
We begin by establishing the relevant symmetries of simplicial complexes, which are described by a group action. Recall the definition of a permutation of a set:

\begin{definition}[\textbf{Permutation, Transposition}]
A \emph{permutation} of a set $X$ is a bijection from $X$ onto itself, and the set of all permutations of $X$ is the symmetric group on $X$ which we denote by $\Sym (X)$. A \emph{transposition} is a permutation that exchanges two elements of $X$ and leaves the other elements fixed.
\end{definition}

\begin{notation}
We write permutations in both function notation and cycle notation, as is convenient, and we compose permutations from right to left. In cycle notation a transposition is a 2-cycle, written $(x\; y)$ for $x$, $y \in X$.
\end{notation}

\begin{remark}
Herein our convention is to view permutations in the active sense.
\end{remark}

Formally we refer to a symmetry of a simplicial complex as an automorphism. Recall the following standard definition:

\begin{definition}[\textbf{Automorphism of a simplicial complex, Automorphism group}]
An \emph{automorphism}, or symmetry, of a simplicial complex $K$ is a permutation on the set of vertices $\Vertices(K)$ that acts simplicially. The \emph{automorphism group} of $K$ is the set of all automorphisms of $K$, denoted $\Aut (K)$, which is isomorphic to a subgroup of $\Sym \big(\Vertices(K)\big)$.
\end{definition}

Herein we denote the set of positive integers from 1 to $n$ as $[n]$.

\begin{notation}
Let $\pi$ be an automorphism of a simplicial complex $K$, and denote $\Vertices (K) = \{v_i\}_{i=1}^n$. If $\sigma := \{v_i\}_{i \in I}$ is a simplex in $K$ for some subset $I \subseteq [n]$, then we denote the action of $\pi$ on $\sigma$ by $\pi(\sigma) = \{\pi(v_i)\}_{i \in I}$. We will usually denote vertices as singleton sets of positive integers, in which case $\Vertices (K) = \big\{\{a_i\}\big\}_{i=1}^n$ where each $a_i$ is a positive integer, and we then denote the action of $\pi$ on $\sigma$ by $\pi(\sigma) = \big\{\{\pi(a_i)\}\big\}_{i \in I}$.
\end{notation}

The automorphisms of a simplicial complex are described by the action of a group on the complex. Our definition of a (left) group action on a simplicial complex is standard \cite[Chapter III, Page 115]{Bredon1972}.

\begin{definition}[\textbf{Group action}]
Let $G$ be a group with identity $e$, and let $X$ be a set. An \emph{action} of the group $G$ on $X$ is a map $\Theta \colon G \times X \to X$ such that:
\begin{enumerate}
\item $\Theta (gh,x) = \Theta (g,\Theta(h,x))$ for all $g$, $h \in G$ and $x \in X$;
\item $\Theta (e,x) = x$ for all $x \in X$.
\end{enumerate}
If $K$ is a simplicial complex then, setting $X = K$, we also assume that:
\begin{enumerate}[resume]
\item For each $g \in G$ the map $\theta_g \colon K \to K$, given by $\theta_g (\sigma) = \Theta (g,\sigma)$ for $\sigma \in K$, acts simplicially on $K$.
\end{enumerate}
We say that each $g \in G$ acts simplicially on $K$, and each such action of a group $G$ is called a \emph{$G$-action}.
\end{definition}

\begin{remark}
Equivalently, an action of the group $G$ on the set $X$ is a group homomorphism $\Phi \colon G \to \Aut (X)$. In particular, recalling that an automorphism of a simplicial complex $K$ is simplicial, a \emph{simplicial action} of the group $G$ on the simplicial complex $K$ is a group homomorphism $\Phi \colon G \to \Aut (K)$.
\end{remark}

\begin{notation}
For an action of the group $G$ on the set $X$, and for $g \in G$ and $x \in X$, we denote $\Theta (g,x) = \Phi (g) (x)$ by $g \cdot x$.
\end{notation}

While conditions of \emph{regularity} \cite[Chapter III, Page 116, Definition 1.2]{Bredon1972} are generally assumed for $G$-actions, such conditions are neither required nor desirable for our purposes since we are interested in identifying labelled simplices that, upon permutation of the spanning vertices, produce an invariant simplicial representation. Indeed, a \emph{regular} simplicial action of the group $G$ on the simplicial complex $K$ is a simplicial action that satisfies the following condition for the action of each subgroup of $G$ \cite[Chapter III, Page 115, Statement (B)]{Bredon1972}: Condition (A) --- if $g_1, g_2, \ldots, g_m \in G$, and $\{ v_1, v_2, \ldots, v_m \}$ and $\{ g_1 \cdot v_1, g_2 \cdot v_2, \ldots, g_m \cdot v_m \}$ are both simplices of $K$, then there exists an element $g \in G$ such that $g \cdot v_i = g_i \cdot v_i$ for all $i \in [m]$. Condition (A) implies the following Condition (B) \cite[Chapter III, Page 116, Statement (A$^{\prime}$) and the following two paragraphs]{Bredon1972}: Condition (B) --- if $g \in G$, $v$ is a vertex in $K$, and $v$ and $g \cdot v$ belong to the same simplex, then $v = g \cdot v$. Condition (B) illustrates that regular actions fix vertices that are sent to the same simplex. To the contrary, we require nonregular actions that allow for the permutation of the vertices of a given simplex, as illustrated in the following simple example.

\begin{example}
Let $K$ be the 1-simplex $\big\{\, \{1\}, \{2\}, \{1,2\} \,\big\}$, regarded as a simplicial representation. Since the vertices $\{1\}$ and $\{2\}$ are adjacent and in symmetrical positions in the complex, they are conceptually equivalent so may be identified by an adjacent-vertex identification. While we can directly observe from $K$ that the two vertices are conceptually equivalent, this is very difficult for more complicated simplicial representations, and the use of group actions simplifies the process. In this case, let $G = \big\langle (1\; 2) \big\rangle$ be the permutation group on the set $\{1,2\}$ generated by the transposition $(1\; 2)$, that is, $G = \big\{e, (1\; 2)\big\}$. Then $G$ acts simplicially on $K$ by permuting the vertices of $K$ (we discuss such actions in detail below). Now, $\{ 1,2 \}$ and $\big\{ e \cdot 1, (12) \cdot 2 \big\} = \{1\}$ are both simplices of $K$, however there is no $g \in G$ such that $g \cdot \{1,2\} = \{1\}$. So $G$ does not satisfy Condition (A), and hence is not a regular action. Such a group action is, however, a symmetry of interest in our work here, since the permutation of the vertices of the simplex $K$ results in an invariant complex.
\end{example}

To reveal any conceptually-equivalent vertices in a simplicial representation we require a specific class of simplicial automorphisms, which we call exchange automorphisms (see \cite[Section 3.2, Page 322, Definition 3.20]{Papadopoulos2012}).

\begin{definition}[\textbf{Exchange automorphism, Exchangeable vertices}]\label{def:Exchange}
Let $K$ be a simplicial complex, and let $u$, $v \in \Vertices (K)$. An \emph{exchange automorphism} for the vertices $u$ and $v$ is an automorphism $\phi \in \Aut (K)$ such that $\phi (u) = v$, $\phi (v) = u$, and $\phi (w) = w$ for all $w \in \Vertices (K) \setminus \{u,v\}$. If an exchange automorphism exists for vertices $u$ and $v$ then we say that $u$ and $v$ are \emph{exchangeable} in $K$.
\end{definition}

Note that if $K$ is a simplicial complex and $u$, $v \in \Vertices (K)$ are exchangeable then the associated exchange automorphism on $K$ is unique.

\begin{remark}\label{rem:ExchangeEquiv}
Let $K$ be a simplicial complex. The binary relation $\{\, (u,v) \mid \text{$u$, $v \in \Vertices (K)$ are exchangeable} \,\}$ over the set of exchangeable vertices in $K$ is an equivalence relation. Reflexivity follows from the identity automorphism on $K$, and symmetry follows from the definition of exchangeable vertices. For transitivity we observe that if $\phi \in \Aut (K)$ is an exchange automorphism for the vertices $u$ and $v$, and $\psi \in \Aut (K)$ is an exchange automorphism for the vertices $v$ and $w$, then $\phi \circ \psi \circ \phi \in \Aut (K)$ is an exchange automorphism for the vertices $u$ and $w$.
\end{remark}

We employ group actions of a particular class:

\begin{definition}[\textbf{Group action by exchangeable vertices}]
A \emph{group action by exchangeable vertices} on a simplicial complex $K$ consists of an action of a group $G$ on $K$ such that $G = \langle M \rangle$ where:
\begin{itemize}
\item $M \subseteq \Sym \big(\Vertices(K)\big)$ is a set of transpositions $(u\;v)$ where $u$ and $v$ are exchangeable in $K$; and,
\item the action of each $(u\;v) \in M$ is the exchange automorphism on $K$ that exchanges $u$ and $v$.
\end{itemize}
\end{definition}

Once identifiable vertices are found in a simplicial representation through a group action, we can obtain a new simplicial complex, called the orbit space, with these vertices identified as single vertices. The definition of the orbit space of an action of a group on a simplicial complex follows \cite[Chapter III, Page 117, Paragraph 2]{Bredon1972}.

\begin{definition}[\textbf{Orbit, Orbit space}]\label{def:Orbit}
Let $K$ be a simplicial complex, and suppose that the group $G$ acts on $K$. The \emph{orbit}, or \emph{$G$-orbit}, of an element $\sigma \in K$ is the set $G \cdot \sigma := \{\, g \cdot \sigma \mid g \in G \,\}$. The \emph{orbit space}, or \emph{quotient}, of an action of a group $G$ on $K$, denoted $K/G$, is the set consisting of the orbits $\Vertices (K/G) := \{\, G \cdot v \mid v \in \Vertices (K) \,\}$, along with the finite subsets $\{ G \cdot v_i \}_{i=1}^{n}$ of $\Vertices (K/G)$ for which $\{ u_i \}_{i=1}^{n}$ spans a simplex in $K$, where each $u_i$ is a representative of $G \cdot v_i$, and the existence of such a simplex in $K$ is not required for all systems of representatives of the orbits $G \cdot v_i$.
\end{definition}

\begin{example}
Let $K$ be the simplicial 2-complex $\big\{ \{1\}, \{2\}, \{3\}, \{1,2\}, \{1,3\}, \{2,3\}, \{1,2,3\} \big\}$, which is geometrically a filled triangle. Let $G := \big\langle (1\;2) \big\rangle \cong \ZZ / 2\ZZ$ be the cyclic subgroup of $S_3$, the symmetric group of degree 3, generated by the transposition $(1\; 2)$. Then $g \cdot K = K$ for all $g \in G$. The action of $G$ on $K$ is illustrated geometrically in Figure~\ref{fig:Fig1}(a), where $(1\; 2)$ is a reflection.
\begin{figure}[!h]
\centering\includegraphics[width=0.9\textwidth]{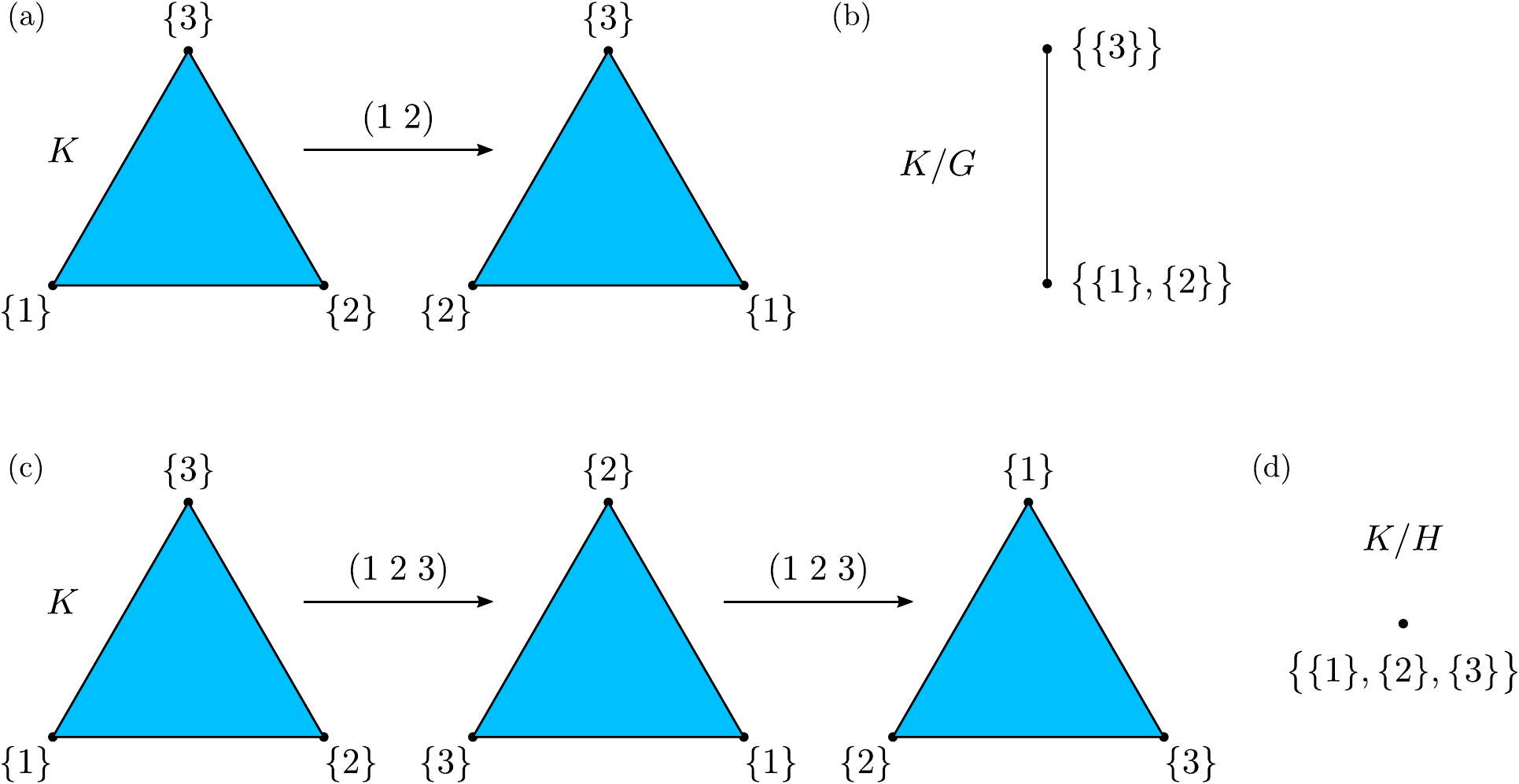}
\caption{Geometric illustration of the action of the group $G$ on the simplicial complex $K$: (a) The action of the transposition $(1\;2)$ on $K$; (b) The orbit space $K/G$ is the 1-simplex with vertices $\big\{\{3\}\big\}$ and $\big\{\{1\},\{2\}\big\}$. Geometric illustration of the action of the group $H$ on the simplicial complex $K$: (c) The action of the permutation $(1\;2\;3)$ on $K$; (d) The orbit space $K/H$ is the 0-simplex $\big\{\{1\},\{2\},\{3\}\big\}$.}
\label{fig:Fig1}
\end{figure}
We have $G \cdot \{1\} = G \cdot \{2\} = \big\{ \{1\}, \{2\} \big\}$ and $G \cdot \{3\} = \big\{ \{3\} \big\}$, hence $\Vertices (K/G) = \big\{ G \cdot \{1\}, G \cdot \{3\} \big\}$. Since $\{1,3\}$, or alternatively $\{2,3\}$, is a simplex in $K$, $\big\{ G \cdot \{1\}, G \cdot \{3\} \big\}$ is a simplex in $K/G$. Therefore $K/G = \Big\{ G \cdot \{1\}, G \cdot \{3\}, \big\{ G \cdot \{1\}, G \cdot \{3\} \big\} \Big\}$ is a 1-simplex, as illustrated in Figure~\ref{fig:Fig1}(b), resulting from the identification of the symmetric vertices $\{1\}$ and $\{2\}$ in $K$.

Now let $H := \big\langle (1\; 2\; 3) \big\rangle \cong \ZZ / 3\ZZ$ be the cyclic subgroup of $S_3$ generated by the cyclic permutation $(1\; 2\; 3)$. Then $h \cdot K = K$ for all $h \in H$. The action of $H$ on $K$ is illustrated geometrically in Figure~\ref{fig:Fig1}(c), where $(1\; 2\; 3)$ is a counter-clockwise rotation. We have $H \cdot \{1\} = H \cdot \{2\} = H \cdot \{3\} = \big\{ \{1\}, \{2\}, \{3\} \big\}$, and it follows that $K/H = \big\{H \cdot \{1\}\big\}$ is a 0-simplex, as illustrated in Figure~\ref{fig:Fig1}(d), resulting from the identification of the three vertices in $K$.
\end{example}

It is a standard observation that the orbit space $K/G$ is a simplicial complex \cite[Chapter III, Page 117, Paragraph 2]{Bredon1972}, however we provide a proof for completeness.

\begin{proposition}
Let $K$ be a simplicial complex, and let $G$ be a group action on $K$. Then the orbit space $K/G$ is a simplicial complex with vertex set $\Vertices (K/G)$.
\end{proposition}

\begin{proof}
By definition, $K/G$ consists of the set of vertices $\Vertices (K/G) := \{\, G \cdot v \mid v \in \Vertices (K) \,\}$ along with the subsets $\{ G \cdot v_i \}_{i \in I}$ of $\Vertices (K/G)$, for index set $I$, for which there exists $( u_i )_{i \in I} \in \prod_{i \in I} G \cdot v_i$ such that $\{ u_i \}_{i \in I}$ spans a simplex in $K$.

To show that $K/G$ is a simplicial complex we need to show that for each subset $\{ G \cdot v_i \}_{i \in I}$ of $\Vertices (K/G)$ that spans a simplex in $K/G$, every nonempty subset of $\{ G \cdot v_i \}_{i \in I}$ also spans a simplex in $K/G$. So let $\{ G \cdot v_j \}_{j \in J}$ be a subset of $\{ G \cdot v_i \}_{i \in I} \in K/G$, where $J \subseteq I$ is nonempty. Since $\{ G \cdot v_i \}_{i \in I} \in K/G$, there exists a simplex $\{ u_i \}_{i \in I} \in K$ with $( u_i )_{i \in I} \in \prod_{i \in I} G \cdot v_i$. Then $\{ u_j \}_{j \in J}$ is a subsimplex of $\{ u_i \}_{i \in I}$ in $K$, with $( u_j )_{j \in J} \in \prod_{j \in J} G \cdot v_j$, hence $\{ G \cdot v_j \}_{j \in J}$ is a simplex in $K/G$, as required.
\end{proof}

We now show that there is a canonical map from a simplicial complex onto the orbit space corresponding to a $G$-action. This simplicial map induces a vertex-identification operation between the two simplicial complexes.

\begin{proposition}\label{prop:projection}
Let $K$ be a simplicial complex, and let $G$ be a group action on $K$. Then the canonical vertex map $v \mapsto G \cdot v$ extends to a surjective simplicial map $p \colon K \to K/G$.
\end{proposition}

\begin{proof}
Let $\psi \colon \Vertices (K) \to \Vertices (K/G)$ be the canonical vertex map such that $v \mapsto G \cdot v$. Define the map $p \colon K \to K/G$ by $p \big(\{ v_i \}_{i \in I} \big) := \big\{ \psi (v_i) \big\}_{i \in I}$ for a simplex $\{ v_i \}_{i \in I} \in K$, noting that the $\psi (v_i)$ may be equal for distinct indices in $I$. Then $p$ is simplicial since, letting $J \subseteq I$ be a subset of smallest cardinality such that $\big\{ \psi (v_j) \big\}_{j \in J} = \big\{ \psi (v_i) \big\}_{i \in I}$, we have that $\{ v_j \}_{j \in J}$ is a simplex in $K$ and hence, by the definition of $K/G$, $\big\{ \psi (v_j) \big\}_{j \in J}$ is a simplex in $K/G$. Surjectivity of $p$ follows from the surjectivity of $\psi$, and $p$ extends $\psi$ since $p |_{\Vertices (K)} = \psi$.
\end{proof}

\begin{notation}
We refer to the canonical map $p \colon K \to K/G$ in Proposition~\ref{prop:projection} as the \emph{projection map}.
\end{notation}

When we identify conceptually-equivalent vertices in a simplicial representation $K$ to give an orbit space, we would expect that the orbit space is isomorphic to a simplicial subcomplex of $K$ as a consequence of the underlying symmetry identified with the group action. We now confirm that this is indeed the case, beginning with two lemmas.

\begin{lemma}\label{lemma:transposition}
Let $G$ be a group acting on a simplicial complex $K$ by exchangeable vertices, and let $u \in \Vertices (K)$. If $v$, $w \in G \cdot u$ then the transposition $(v\;w)$ is in $G$. In particular, $v$ and $w$ are exchangeable in $K$.
\end{lemma}

\begin{proof}
Since $v$, $w \in G \cdot u$ there exists $g \in G$ such that $g \cdot v = w$, where $g = t_n t_{n-1} \cdots t_2 t_1$ is a product of transpositions $t_i \in G$ for $i \in [n]$. Without loss of generality we may assume that there is no proper subsequence of the sequence of transpositions $(t_i)_{i=1}^{n}$ whose product sends $v$ to $w$. In particular: $v$ is an element of $t_1$ only; $w$ is an element of $t_n$ only; and the transpositions are mutually disjoint unless they are successive in the sequence, whereby their intersection is a singleton. We then have $t_1 = (v\;x_1)$, $t_{i+1} = (x_i\;x_{i+1})$ for $i \in [n-2]$, and $t_n = (x_{n-1}\;w)$, for some $\{x_i\}_{i=1}^{n-1} \subseteq \Vertices (K)$. It follows that $(v\;w) = t_1 t_2 \cdots t_{n-1} t_n t_{n-1} \cdots t_2 t_1 \in G$.
\end{proof}

\begin{lemma}\label{lemma:IsoSim}
Let $G$ be a group acting on a simplicial complex $K$ by exchangeable vertices. Suppose that $\{ G \cdot v_i \}_{i \in I}$ spans a simplex in $K/G$, for some index set $I$, and $\{ u_i \}_{i \in I}$ spans a simplex in $K$ with $(u_i)_{i \in I} \in \prod_{i \in I} G \cdot v_i$. Then for each $(w_i)_{i \in I} \in \prod_{i \in I} G \cdot v_i$ we have that $\{ w_i \}_{i \in I}$ spans a simplex in $K$, and there is a $g \in G$ such that, for each $i \in I$, $g \cdot u_i = w_i$, $g \cdot w_i = u_i$, and $g \cdot v = v$ for all $v \in \Vertices (K) \setminus \big( \{ u_i \}_{i \in I} \cup \{ w_i \}_{i \in I} \big)$.
\end{lemma}

\begin{proof}
Suppose $(w_i)_{i \in I} \in \prod_{i \in I} G \cdot v_i$. It follows from Lemma~\ref{lemma:transposition} that for each $i \in I$ the transposition $t_i := (u_i\; w_i)$ is in $G$. So define $g \in G$ as the product $g := \prod_{i \in I} t_i$, where the order in which the $t_i$ are multiplied is unimportant since these group elements mutually commute: indeed, two vertices from distinct $G$-orbits are not exchangeable. Then $g \in G$ is the required group element, and it follows that the action of $g$ maps the simplex spanned by $\{ u_i \}_{i \in I}$ to the simplex spanned by $\{ w_i \}_{i \in I}$.
\end{proof}

\begin{definition}[\textbf{Fundamental domain}]
Let $K$ be a simplicial complex with an action of the group $G$ on $K$. A \emph{fundamental domain} for the action of $G$ on $K$ is a full subcomplex $L$ of $K$ that intersects each vertex $G$-orbit exactly once.
\end{definition}

Note that a fundamental domain is not necessarily unique, and the group action transforms between the different fundamental domains of the complex. We now show that a fundamental domain exists for a group action on a simplicial complex by exchangeable vertices, and that the corresponding orbit space is isomorphic to the fundamental domain. From this result one infers that all fundamental domains are isomorphic for a given group action on a simplicial complex by exchangeable vertices.

\begin{theorem}
Let $G$ be a group acting on a simplicial complex $K$ by exchangeable vertices. Then a fundamental domain exists for the action of $G$ on $K$, and the orbit space $K/G$ is isomorphic to the fundamental domain.
\end{theorem}

\begin{proof}
We first show that a fundamental domain exists for $G$ on $K$. Let $\Vertices (K/G) := \{\, G \cdot v_i \,\}_{i \in I}$ be the set of distinct vertex orbits where $I$ is the index set, choose $( w_i )_{i \in I} \in \prod_{i \in I} G \cdot v_i$, and let $L$ be the full subcomplex of $K$ with vertex set $\{ w_i \}_{i \in I}$. Then $L$ is a fundamental domain.

It remains to show that $K/G$ is isomorphic to $L$. Let $\lambda \colon K/G \to L$ be the map such that if $\sigma$ is a simplex in $K/G$ spanned by the vertices $\{\, G \cdot v_j \,\}_{j \in J}$, where $J \subseteq I$, then $\lambda (\sigma)$ is the simplex in $L$ spanned by the vertices $\{ w_j \}_{j \in J}$. The map $\lambda$ is well defined, since if the simplex $\sigma \in K/G$ is spanned by $\{\, G \cdot v_j \,\}_{j \in J}$ then there exists a simplex in $K$ spanned by a set of vertices $\{ u_j \}_{j \in J}$ with $( u_j )_{j \in J} \in \prod_{j \in J} G \cdot v_j$, so by Lemma~\ref{lemma:IsoSim} the vertices $\{ w_j \}_{j \in J}$ span a simplex in $K$, which is also in the full subcomplex $L$. The map $\lambda$ is then a simplicial bijection, and hence a simplicial isomorphism.
\end{proof}

Since the orbit space may have exchangeable vertices that were either not accounted for or not present in the original complex, we can obtain the orbit space of the orbit space, and so on for a finite number of steps until an orbit space is obtained which has no further exchangeable vertices. To see this more clearly, we now show how a group action on a simplicial complex by exchangeable vertices relates to the full group action on the corresponding orbit space by exchangeable vertices. Note that we say that a group action on a simplicial complex by exchangeable vertices is a \emph{full group action} when the group action gives all possible exchange automorphisms of the complex. We first require a lemma.

\begin{lemma}\label{lemma:Ex}
Let $G$ be a group action on a simplicial complex $K$ by exchangeable vertices, and let $t$ be a transposition, not necessarily in $G$, from the full group action on $K$ by exchangeable vertices that acts to exchange the vertices $u$, $w \in \Vertices (K)$. Further, let $H$ be the full group action on $K/G$ by exchangeable vertices, and let $\{ G \cdot v_i \}_{i \in I}$ be the set of distinct orbits in $K/G$ which partition $\Vertices (K)$. Then there exists a transposition $\overline{t} \in H$ that acts to exchange the vertices $G \cdot v_j$, $G \cdot v_k \in \Vertices (K/G)$ where $j$, $k \in I$, $u \in G \cdot v_j$, and $w \in G \cdot v_k$. Further, $\overline{t}$ is the identity transposition if and only if $t \in G$.
\end{lemma}

\begin{proof}
Suppose that $t \in \Sym \big(\Vertices (K)\big)$ transposes the vertices $u$, $w \in \Vertices (K)$, and that $(u,w) \in G \cdot v_j \times G \cdot v_k$ for some $j$, $k \in I$. The permutation $\pi_t$ on the vertices $\{ G \cdot v_i \}_{i \in I}$ of $K/G$ which transposes $G \cdot v_j$ and $G \cdot v_k$ and fixes all other vertices induces a map $\phi_t \colon K/G \to K/G$ which sends a simplex $\tau \in K/G$ spanned by $\{ G \cdot v_i \}_{i \in J}$ for some $J \subseteq I$ to the simplex $\phi_t (\tau)$ spanned by $\big\{ \pi_t (G \cdot v_i) \big\}_{i \in J}$ where $G \cdot v_j$ and $G \cdot v_k$ are exchanged, if required.

To show that $\phi_t$ is well defined we need to establish that $\phi_t (\tau)$ is a simplex in $K/G$. Since $\{ G \cdot v_i \}_{i \in J}$ spans the simplex $\tau \in K/G$, there exists $(u_i)_{i \in J} \in \prod_{i \in J} G \cdot v_i$ such that $\{ u_i \}_{i \in J}$ spans a simplex $\sigma \in K$ by the definition of $K/G$. By Lemma~\ref{lemma:IsoSim} we may assume without loss of generality that $u_j = u$ if $j \in J$ and $u_k = w$ if $k \in J$. Then the image of $\sigma$ under the action of $t$ is the simplex $t \cdot \sigma \in K$ spanned by $\{ t \cdot u_i \}_{i \in J}$ where $u$ and $w$ are exchanged, if required. Now, if $i \ne j$ or $k$ then $t \cdot u_i = u_i \in G \cdot v_i = \pi_t (G \cdot v_i)$, if $i = j$ then $t \cdot u_i = t \cdot u_j = u_k \in \pi_t (G \cdot v_j)$, and if $i = k$ then $t \cdot u_i = t \cdot u_k = u_j \in \pi_t (G \cdot v_k)$. It follows that the simplex $t \cdot \sigma \in K$ satisfies $(t \cdot u_i)_{i \in J} \in \prod_{i \in J} \pi_t (G \cdot v_i)$, hence $\phi_t (\tau)$ is a simplex in $K/G$. So $\phi_t$ is a well-defined simplicial map.

Since $\phi_t$ is self-inverse it is bijective and therefore an automorphism. It follows that $\phi_t$ is an exchange automorphism that exchanges the vertices $G \cdot v_j$ and $G \cdot v_k$, so the transposition $\overline{t} := (G \cdot v_j \; G \cdot v_k)$ is in $H$.

Finally, $t \in G$ if and only if $u$ and $w$ are in the same $G$-orbit if and only if $G \cdot v_j = G \cdot v_k$ if and only if $\overline{t}$ is the identity transposition, where the first equivalence uses Lemma~\ref{lemma:transposition}.
\end{proof}

\begin{proposition}\label{prop:AutOrb}
Let $G$ and $F$ be two group actions on a simplicial complex $K$ by exchangeable vertices, and let $H$ be the full group action on $K/G$ by exchangeable vertices. Then there is a group homomorphism $\alpha \colon F \to H$ such that $\alpha (t_n \cdots t_2 t_1)  = \overline{t}_n \cdots \overline{t}_2 \overline{t}_1$ for each element $t_n \cdots t_2 t_1 \in F$, where the $t_i$ for $i \in [n]$ are generating transpositions in $F$, and $\ker (\alpha) = F \cap G$.
\end{proposition}

\begin{proof}
We first show that $\alpha \colon F \to H$ is a well-defined map. Let $f \in F$, and suppose that $f = t_m \cdots t_2 t_1$ where the $t_i$, for $i \in [m]$, are generating transpositions in $F$. By Lemma~\ref{lemma:Ex}, each $\overline{t}_i$ for $i \in [m]$ is a transposition in $H$, so $\overline{f} := \overline{t}_m \cdots \overline{t}_2 \overline{t}_1 \in H$. Now, if we also have that $f = s_n \cdots s_2 s_1$ where the $s_i$, for $i \in [n]$, are generating transpositions in $F$, then the corresponding permutations of $\Vertices (K/G)$, namely $\overline{t}_m \cdots \overline{t}_2 \overline{t}_1$ and $\overline{s}_n \cdots \overline{s}_2 \overline{s}_1$, are equal. It follows that $\alpha$ is well defined.

To show that $\alpha$ is a homomorphism, let $t_m \cdots t_2 t_1$ and $s_n \cdots s_2 s_1$ be two permutations in $F$, in terms of the generating transpositions. Then $\alpha \big((t_m \cdots t_2 t_1) (s_n \cdots s_2 s_1)\big) = (\overline{t}_m \cdots \overline{t}_2 \overline{t}_1) (\overline{s}_n \cdots \overline{s}_2 \overline{s}_1) = \alpha(t_m \cdots t_2 t_1) \alpha(s_n \cdots s_2 s_1)$, and the result follows.

It remains to show that the kernel of $\alpha$ is equal to $F \cap G$. For this, let $\{ G \cdot v_i \}_{i \in I}$ be the set of distinct vertices in $K/G$, which partition $\Vertices (K)$. Suppose that $f \in \ker (\alpha)$, so that $\alpha (f) = \overline{f}$ is the identity in $H$. Then $\overline{f}$ fixes each vertex $G \cdot v_i$ of $K/G$, so the definition of $\overline{f}$ implies that $f \cdot (G \cdot v_i) \subseteq G \cdot v_i$ for $i \in I$. We can write $f$ as a product of disjoint permutations, $f = \prod_{i \in I} p_i$, where each such permutation $p_i$ permutes only elements in $G \cdot v_i$ with the same action as $f$ on $G \cdot v_i$, and $p_i$ fixes all elements in $\Vertices (K) \setminus G \cdot v_i$. Then each $p_i$ can be expressed as a product of transpositions in $G$, hence $f \in G$. It follows that $\ker (\alpha) \subseteq F \cap G$. Conversely, if $f \in F \cap G$ then, since $f \in G$, $\alpha (f) = \overline{f} \in H$ fixes the $G$-orbits of $\Vertices (K)$, which are the vertices of $K/G$, so $\alpha (f)$ is the identity in $H$, hence $f \in \ker (\alpha)$. Therefore $F \cap G \subseteq \ker (\alpha)$.
\end{proof}

In Subsection~\ref{subsec:method} we employ the theory developed here to provide a methodology for determining model equivalence.

\subsection{Methodology for determining model equivalence}\label{subsec:method}
To compare the simplicial representations corresponding to two models, with respect to equivalence, we can compare the two corresponding orbit spaces that have no further exchangeable vertices, which we refer to as the \emph{final} orbit spaces. These orbit spaces, which may be smaller than the original simplicial complexes, allow for the identification of subsets of equivalent components within each model, and in turn the equivalence of components between the two models.

Employing the theory in Subsection~\ref{subsec:Gaction} for identifying equivalent components within simplicial representations, we now provide a general methodology for determining model equivalence in terms of Operations 1, 2, and 5. There is not necessarily a unique order with which Operations 1, 2, and 5 should be applied to a simplicial representation. Given two simplicial representations, however, the easiest approach is to first reduce the simplicial representations by applying vertex identifications (Operations 1 and 2) when possible, and finally relabelling one of the complexes by a vertex substitution (Operation 5). We will assume that all possible vertex identifications for a given simplicial representation are applied simultaneously. We could apply a proper subset of the possible vertex identifications to the simplicial representation, and then apply the remaining identifications to the corresponding orbit space (see Proposition~\ref{prop:AutOrb}), however applying them all at once minimises the number of required simplicial maps.

The following steps describe the process of finding model equivalences using simplicial representations. For the purpose of automation on computer we can relabel the vertices of the simplicial representations, which correspond to model components, with positive integers so that the simplices of the simplicial representations are subsets of these positive integers. Assume that $K$ and $L$ are two simplicial representations corresponding to two models.

\begin{enumerate}[label=\textbf{Step \arabic*.},leftmargin=!,labelindent=\parindent]
\item \textbf{Find all pairs of exchangeable vertices of $K$ and $L$:} We describe the process for $K$, and the process for $L$ is similar. Two vertices $u$ and $v$ in $\Vertices (K)$ are \emph{exchangeable} if exchanging the labels of the two vertices leaves $K$ unchanged (Definition~\ref{def:Exchange}). All exchangeable vertices of $K$ can be found by exchanging all pairs of vertices of $K$ in succession, however this process can be made more efficient in a number of ways: for example, exchangeable vertices must have the same number of adjacencies, so two vertices with different numbers of adjacencies are not exchangeable. Given a pair of exchangeable vertices $u$ and $v$ in $\Vertices (K)$ there is a corresponding transposition $(u\;v) \in \Sym \big(\Vertices (K)\big)$, and the set of all such transpositions corresponding to exchangeable vertices in $K$ generates a subgroup $G_1$ of $\Sym \big(\Vertices (K)\big)$. We then have an action of $G_1$ on $K$ by automorphisms, from which we construct the orbit space $K/G_1$ with vertices the $G_1$-orbits of the vertices in $K$ (Definition~\ref{def:Orbit}). If $K$ has no exchangeable vertices then $G_1$ is the trivial subgroup which acts as the identity automorphism of $K$, so we can then take $K/G_1$ to be $K$ since these complexes are isomorphic.

A transposition $(u\;v) \in G_1$ acts as an exchange automorphism of $K$ that exchanges the vertices $u$, $v \in K$, and induces a corresponding adjacent-/nonadjacent-vertex identification operation in terms of a surjective simplicial map $p \colon K \to K/\big\langle (u\;v) \big\rangle$ which sends both $u$ and $v$ to a single vertex in $K/\big\langle (u\;v) \big\rangle$ and fixes all other vertices in $K$. Note that $p$ is the projection map (Proposition~\ref{prop:projection}). Applying all permutations in $G_1$ simultaneously corresponds to $G_1$ acting on $K$, and therefore induces a finite number of vertex-identification operations in terms of the projection map $p_1 \colon K \to K/G_1$.
\item \textbf{Find the sequences of orbit spaces corresponding to $K$ and $L$:} Again, we describe the process for $K$, and it is similar for $L$. Finding the exchangeable vertices of $K/G_1$ follows a process analagous to that for $K$ in Step 1, except the vertices of $K/G_1$ are subsets of vertices of $K$. Given a pair of exchangeable vertices $u$ and $v$ in $\Vertices (K/G_1)$ there is a corresponding transposition $(u\;v) \in \Sym \big(\Vertices (K/G_1)\big)$, and the set of all such transpositions corresponding to exchangeable vertices in $K/G_1$ generates a subgroup $G_2$ of $\Sym \big(\Vertices (K/G_1)\big)$. We then have an action of $G_2$ on $K/G_1$ by automorphisms, from which we construct the orbit space $(K/G_1)/G_2$ with vertices the $G_2$-orbits of the vertices in $K/G_1$. The action of $G_2$ on $K/G_1$ induces a finite number of vertex-identification operations in terms of the projection map $p_2 \colon K/G_1 \to (K/G_1)/G_2$.

Continuing in this manner until we obtain an orbit space with no exchangeable vertices, namely the final orbit space, results in the following sequence of simplicial complexes and projection maps $p_i$ representing vertex identifications,
\begin{equation}\label{eq:Kseq}
K \xrightarrow{\enspace p_1 \enspace} K/G_1 \xlongrightarrow{\enspace p_2 \enspace} (K/G_1)/G_2 \xlongrightarrow{\enspace p_3 \enspace} \cdots \xlongrightarrow{\enspace p_m \enspace} \Big( \cdots \big((K/G_1)/G_2\big) \cdots \Big)/G_m,
\end{equation}
where the final orbit space in the sequence, which we denote by $\widehat{K}$, has no exchangeable vertices. Similarly, for the simplicial representation $L$ we have a sequence of subgroups $H_i$ generated by transpositions giving the following sequence of simplicial complexes and simplicial maps $q_i$ representing vertex identifications,
\begin{equation}\label{eq:Lseq}
L \xrightarrow{\enspace q_1 \enspace} L/H_1 \xrightarrow{\enspace q_2 \enspace} (L/H_1)/H_2 \xrightarrow{\enspace q_3 \enspace} \cdots \xrightarrow{\enspace q_n \enspace} \Big( \cdots \big((L/H_1)/H_2\big) \cdots \Big)/H_n,
\end{equation}
where the final orbit space in the sequence, which we denote by $\widehat{L}$, has no exchangeable vertices.

Note that if any complex from the sequence in Equation~\eqref{eq:Kseq} is isomorphic to a complex from the sequence in Equation~\eqref{eq:Lseq}, as unlabelled simplicial complexes, then $\widehat{K}$ and $\widehat{L}$ are isomorphic as unlabelled simplicial complexes. Therefore, it suffices to determine whether or not $\widehat{K}$ and $\widehat{L}$ are isomorphic, which are the simplest complexes in their respective sequences.

With regard to the sequences in Equations~\eqref{eq:Kseq} and \eqref{eq:Lseq} note that the number of vertices in an orbit space is reduced, relative to the previous complex in the sequence, by the number of distinct and nontrivial transpositions in the associated group action. This observation provides a measure of the reduction in the number of possible isomorphisms between the orbit spaces of each simplicial complex.
\item \textbf{Determine all isomorphisms of $\widehat{K}$ and $\widehat{L}$:} An isomorphism between $\widehat{K}$ and $\widehat{L}$ is a bijective simplicial map from $\widehat{K}$ onto $\widehat{L}$, that is, a bijective vertex map from $\Vertices(\widehat{K})$ onto $\Vertices(\widehat{L})$ that sends a simplex in $\widehat{K}$ to a simplex in $\widehat{L}$. We begin by finding all bijective vertex maps from $\Vertices(\widehat{K})$ onto $\Vertices(\widehat{L})$, noting that there may be none. If there are multiple bijective vertex maps then we may reduce the number of these maps by considering properties that must be preserved by a simplicial isomorphism, such as the dimension of simplices and the number of edge adjacencies of vertices. Further, if $\widehat{K}$ and $\widehat{L}$ each have a vertex label containing the same model component then we would generally ensure that the corresponding vertices must be identified by any bijective vertex map, based on the requirement for conceptual equivalence. We could omit the preservation of model components for the bijective vertex maps, however this will likely only increase the number of bijective vertex maps to consider without yielding any new meaningful equivalences between the model concepts. It should not be the case that a vertex label in $\widehat{K}$ has model components in common with two distinct vertices in $\widehat{L}$, as this would indicate a conceptual inconsistency with at least one of the simplicial representations.

Given a bijective vertex map, which we can describe by a correspondence between the positive-integer labels of the vertices in $\widehat{K}$ and $\widehat{L}$, we then try to extend it to a bijection from $\widehat{K}$ onto $\widehat{L}$ by relabelling $\widehat{K}$ with the corresponding vertex labels of $\widehat{L}$ in accordance with the bijective vertex maps. We then check whether the relabelled complex $\widehat{K}$ is the same as $\widehat{L}$, and if so then the bijective vertex map extends canonically to a simplicial isomorphism, so that $\widehat{K}$ and $\widehat{L}$ are isomorphic. If no simplicial isomorphism exists between $\widehat{K}$ and $\widehat{L}$ then they are not isomorphic as complexes.
\item \textbf{Equivalence of $K$ and $L$:} If $K$ and $L$ are equivalent then $\widehat{K}$ and $\widehat{L}$ are isomorphic; therefore, if $\widehat{K}$ and $\widehat{L}$ are not isomorphic then $K$ and $L$ are not equivalent. If $\widehat{K}$ and $\widehat{L}$ are isomorphic then for each simplicial isomorphism we can consider, for each pair of vertices identified between $\widehat{K}$ and $\widehat{L}$ via the isomorphism, whether the set of model components associated with the vertex in $\widehat{K}$ is conceptually related to the set of model components associated with the vertex in $\widehat{L}$. Since there may be more than one isomorphism between $\widehat{K}$ and $\widehat{L}$ we consider the vertex identifications for each such isomorphism. Importantly, an isomorphism between $\widehat{K}$ and $\widehat{L}$ can reveal novel conceptual relations between $K$ and $L$, and therefore the corresponding models.

If there is an isomorphism between $\widehat{K}$ and $\widehat{L}$ that provides a conceptually meaningful relationship between the concepts of the two models then we can use a vertex-substitution operation induced by the isomorphism to transform $\widehat{K}$ into $\widehat{L}$, and then we conclude that $K$ and $L$, and hence the corresponding models, are equivalent. Otherwise, if no isomorphism between $\widehat{K}$ and $\widehat{L}$ provides a conceptually meaningful relationship between the concepts of the two models then we conclude that $K$ and $L$, and hence the corresponding models, are not equivalent.
\item \textbf{Restricting conceptual equivalence:} A simplicial isomorphism between the final orbit spaces $\widehat{K}$ and $\widehat{L}$ provides for the maximum conceptual equivalence between the components of the two models. It may happen, however, that we regard a certain equivalence for a model, corresponding to a vertex identification, as invalid. In this case, we don't consider $\widehat{K}$ and $\widehat{L}$ as equivalent. There may, however, be equivalent simplicial subcomplexes of $K$ and $L$ with respect to a restricted set of vertex identifications corresponding to the valid conceptual equivalences of model components. To determine whether equivalent subcomplexes exist in this restricted sense, we can work backwards from $\widehat{K}$ and $\widehat{L}$ to remove all of the invalid vertex identifications, while retaining the isomorphic correspondence. If we can remove corresponding vertex identifications stepwise from both $\widehat{K}$ and $\widehat{L}$, where each vertex identification is either adjacent for both complexes or nonadjacent for both complexes, and arrive at isomorphic complexes with only valid equivalences between model components then $K$ and $L$ are equivalent, and otherwise not.
\end{enumerate}

\begin{remark}
Here we briefly discuss some computational aspects of our model comparison methodology.
\begin{enumerate}
\item Orbit spaces: Constructing the orbit space of a simplicial complex $K$ requires first determining the exchange automorphisms, which only involves exchanging pairs of vertices in $K$ and checking that $K$ is unchanged. Once all relevant exchange automorphisms are found, the orbit space is constructed from the vertex orbits and the known simplices in $K$.
Isomorphism of simplicial complexes: We need to determine whether the final orbit spaces are isomorphic. In general, checking isomorphism of simplicial complexes is computationally intensive, however our simplicial complexes are labelled (with sets of model components, or corresponding positive integers) and we are usually interested in isomorphisms that preserve model components. This restriction reduces the number of possible isomorphisms to consider.
Automation: Our methodology for model comparison by equivalence consists of three main parts. Given a model, the first part requires the construction of the simplicial representation. This requires a translation of the model into its underlying concepts and their interrelations. In general this is very difficult to automate, as it requires the ability to identify the model concepts and interrelations for a model with any formalism, as well as equivalences between model concepts. The latter is also dependent on context and individual requirements.

\item The second part of our methodology, which we discuss in this article, is completely automatable up to and including the determination of whether the final orbit spaces are isomorphic. This provides for possible equivalences between model concepts. The second part is the most challenging of all three parts, since establishing whether vertices in a simplicial representation are in symmetrical positions is difficult for even small models.

\item The third part requires determining whether the possible equivalences between model concepts are valid, which returns to part one, and is difficult to automate in general.
\end{enumerate}
\end{remark}

\subsection{Example}\label{subsec:example}
In this example we apply our model-comparison methodology to the two main categories of models for developmental pattern formation, namely Turing-pattern (TP) models and positional-information (PI) models. We have described models from these two categories previously \cite{Vittadello2021b}: four TP models, namely activator-inhibitor, substrate depletion, inhibition of an inhibition, and modulation; five PI models, namely linear gradient, synthesis-diffusion-degradation (SDD), opposing gradients, annihilation, and induction-contraction (active modulation). We also discuss these models in the Appendix (Section~\ref{sec:Appendix}). Interestingly, we found that the TP activator-inhibitor model is equivalent to the PI annihilation model from a significant conceptual perspective. This finding was obtained by visual inspection of the two simplicial representations, quite by chance as it is generally difficult to compare simplicial representations visually. Employing the automatable aspect of our methodology for model equivalence as described in this article, however, provides simple and rigorous comparison of simplicial representations, and we now compare all nine models for equivalence.

Of the nine models for developmental pattern formation we consider here, consisting of four TP models and five PI models, three of the models have final orbit spaces that are not isomorphic to the final orbits spaces of any of the other eight models, so each of these three models is not equivalent to any of the other eight models. These three models are PI induction-contraction, TP inhibition of an inhibition, and TP modulation. We therefore consider the remaining six models. Note that the labelling for these simplicial representations, both as positive integers and model components, is described in Vittadello and Stumpf \cite{Vittadello2021b}. For each of these simplicial representations the 0- and 1-simplices are specified by the model, and the higher-dimensional simplices are obtained by forming cliques, where possible, incrementally in dimensions two and higher. For simplicity, in Figures 2 and 3 we show only the 1-skeletons of the simplicial representations.

The three PI models linear gradient, synthesis-diffusion-degradation, and opposing gradients, all have final orbit spaces that are a 0-simplex, so consist of a single vertex (Figure~\ref{fig:Fig2}).
\begin{figure}
\centering\includegraphics[width=1\textwidth]{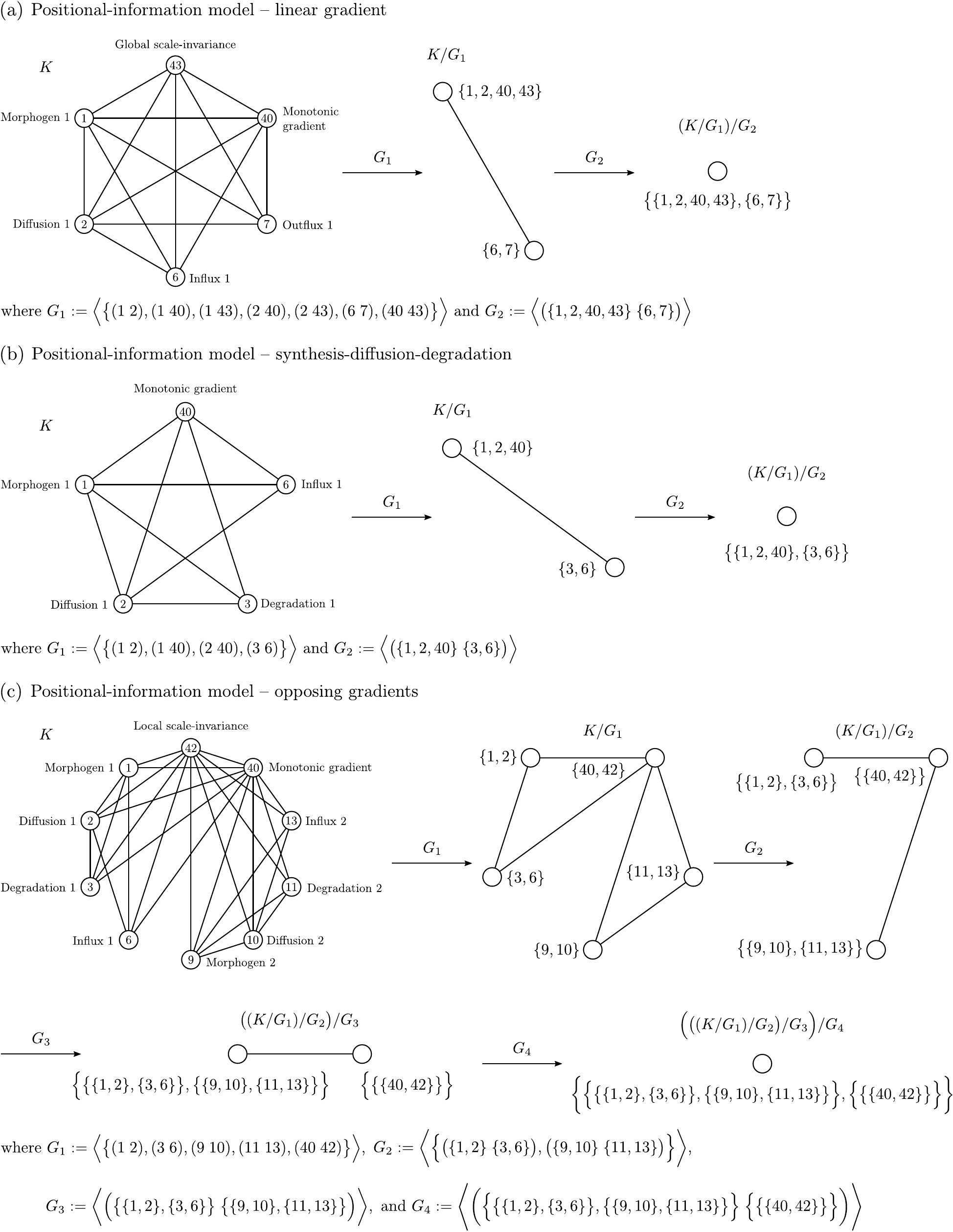}
\caption{1-skeletons of the simplicial representations and orbit spaces for the (a) PI linear gradient, (b) PI synthesis-diffusion-degradation, and (c) PI opposing gradients models.}
\label{fig:Fig2}
\end{figure}
The unlabelled final orbit spaces are therefore isomorphic, so we can compare the corresponding sets of model components to determine model equivalence.

The orbit space for the PI linear gradient model has the label $\big\{\{1,2,40,43\},\{6,7\}\big\}$, and the orbit space for the PI SDD model has the label $\big\{\{1,2,40\},\{3,6\}\big\}$. Equivalence of these two models requires that we consider the model components in $\{1,2,40,43\}$ and $\{6,7\}$ as equivalent to the model components in $\{1,2,40\}$ and $\{3,6\}$, respectively. That is, we need the following: 7 is equivalent to 3; since 1, 2, 40, and 43 are conceptually equivalent, any one of the vertex identifications where 1 and 43 are identified as 1, or 2 and 43 are identified as 2, or 40 and 43 are identified as 40. The equivalence of 7 (Outflux 1) and 3 (Degradation 1) may be reasonable, since they both represent removal of Morphogen 1 from the system. The identification of 43 (Global scale-invariance) with one of 1 (Morphogen 1), 2 (Diffusion 1), or 40 (Monotonic gradient), however, assumes that we are not considering the property of scale-invariance of the morphogen gradient for any model for comparison, which would depend on the required level of conceptual detail. For our purposes, scale-invariance of the morphogen gradient is important for establishing developmental patterning, so is a necessary conceptual detail in the models, and we therefore conclude that the PI linear gradient model is not equivalent to the PI SDD model.

While the PI linear gradient model and PI SDD model each have a single morphogen, the PI opposing gradients model has two morphogens which we regard as an important conceptual detail, so we do not consider the PI opposing gradients model as equivalent to the PI linear gradient model or PI SDD model. Indeed, the label for the orbit space of the PI opposing gradients model shows that the model concepts associated with 1 (Morphogen 1), that is $\big\{\{1,2\},\{3,6\}\big\}$, are identified with the model concepts associated with 9 (Morphogen 2), that is $\big\{\{9,10\},\{11,13\}\big\}$, to effectively reduce to a single morphogen.

In Figure~\ref{fig:Fig3} we consider the PI annihilation model, the TP activator-inhibitor model, and the TP substrate depletion model.
\begin{figure}
\centering\includegraphics[width=1\textwidth]{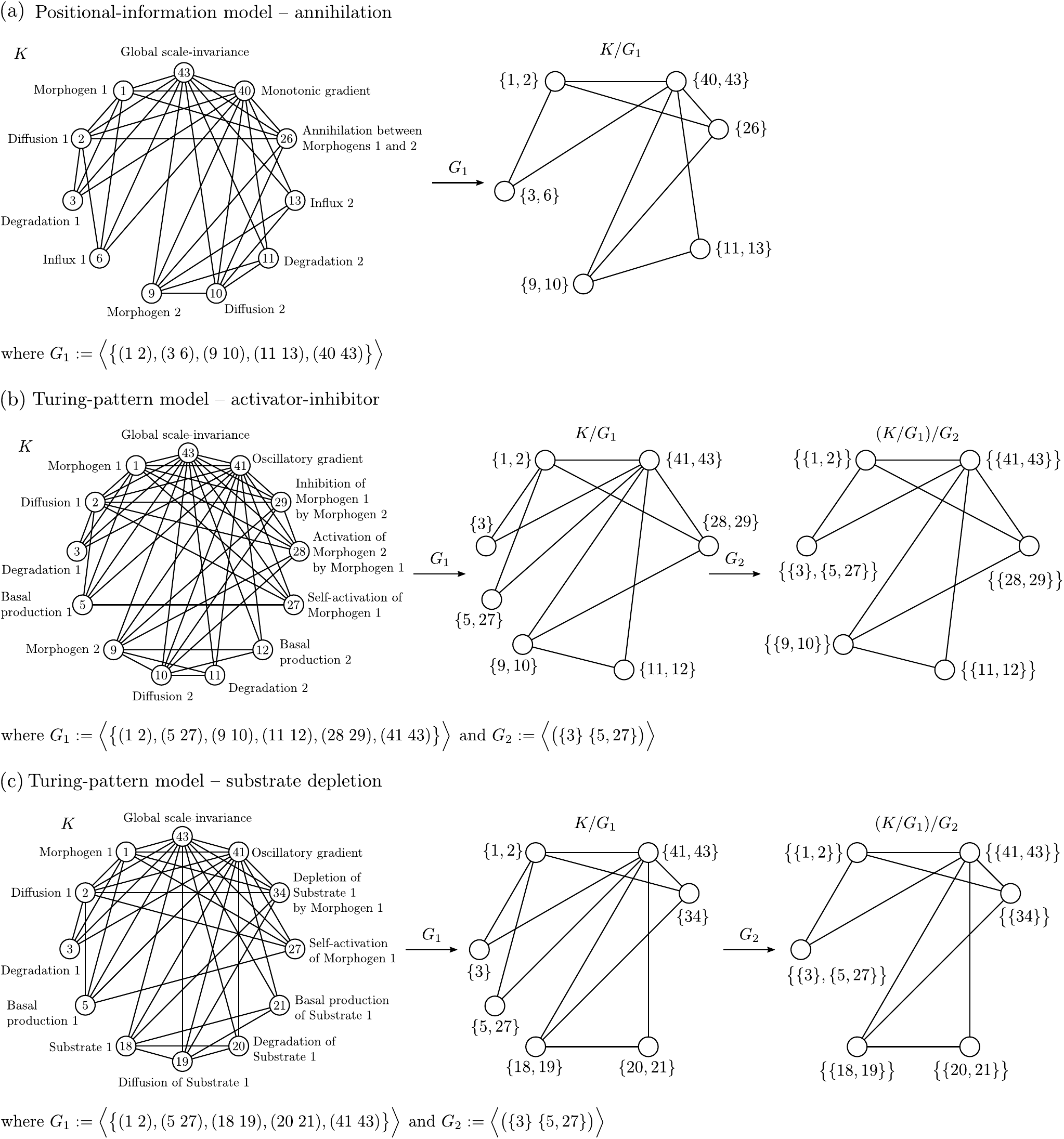}
\caption{1-skeletons of the simplicial representations and orbit spaces for the (a) PI annihilation, (b) TP activator-inhibitor, and (c) TP substrate depletion models.}
\label{fig:Fig3}
\end{figure}
These three models have isomorphic unlabelled final orbit spaces, so we can compare the corresponding sets of model components to determine model equivalence. Note that there are two possible isomorphisms between any two of these three orbit spaces when unlabelled, however only one of these isomorphisms preserves the labelling when the same label is in both complexes. First we consider the equivalence of the PI annihilation model and the TP activator-inhibitor model, which we established in Vittadello and Stumpf \cite{Vittadello2021b}, however now we employ the automatable approach using orbit spaces. We regard the model components 6, 13, 26, and 40 as conceptually equivalent to the model components $\{5,27\}$, 12, $\{28,29\}$, and 41, respectively. In particular, we consider vertex identifications where 5 and 27 are identified as 6, and 28 and 29 are identified as 26. We conclude that the PI annihilation model and the TP activator-inhibitor model are equivalent from a significant conceptual perspective. For further discussion of this example, including a more detailed consideration of the required partial operations, see \cite[Subsubsection 3.2.4]{Vittadello2021b}.

Now consider the equivalence of the TP activator-inhibitor model and the TP substrate depletion model, which would require that the model components 9 (Morphogen 2), 10 (Diffusion 2), 11 (Degradation 2), 12 (Basal production 2), and $\{28,29\}$ (Activation of Morphogen 2 by Morphogen 1, Inhibition of Morphogen 1 by Morphogen 2) are conceptually equivalent to the model components 18 (Substrate 1), 19 (Diffusion of Substrate 1), 20 (Degradation of Substrate 1), and 21 (Basal production of Substrate 1), 34 (Depletion of Substrate 1 by Morphogen 1), respectively. We consider that the agent Morphogen 2 in the TP activator-inhibitor model is conceptually equivalent to the agent Substrate 1 in the TP substrate depletion model. Further, the two reactions given corresponding to 28 and 29 form an inhibitory cycle in the TP activator-inhibitor model, which we consider to be conceptually equivalent to the reaction corresponding to 20 in the TP substrate depletion model (compare with the discussion in Vittadello and Stumpf \cite{Vittadello2021b} regarding the equivalence of the PI annihilation model and the TP activator-inhibitor model). We conclude that the TP activator-inhibitor model and the TP substrate depletion model are equivalent from a significant conceptual perspective, which is not surprising since the TP substrate depletion model is based very closely on the TP activator-inhibitor model. Since the relation of model equivalence is transitive, it also follows that the PI annihilation model is equivalent to the TP substrate depletion model.

\subsection{$G$-representations}\label{subsec:Grep}
In this subsection we provide an overview of an alternative mathematical framework for our model-equivalence methodology by associating a group to a simplicial representation. While the group-theoretic framework is closely related to the simplicial-complex framework for model equivalence, the group-theoretic approach provides an alternative mathematical perspective that allows for the application of group-theoretic techniques for model comparison. While the mathematical details of the two frameworks differ, the underlying methodology is the same, so we provide only the minimal details required to develop the group-theoretic framework for model equivalence.

We first describe a standard construction of a group from a given set.

\begin{proposition}\label{prop:PowerSetGroup}
Let $S$ be a set and let $2^S$ be the power set of $S$. With the closed binary operation of symmetric difference, $\triangle \colon 2^S \times 2^S \to 2^S$, the pair $(2^S,\triangle)$ is a group.
\end{proposition}

\begin{proof}
Under the operation $\triangle$, $2^S$ is closed, the identity is the empty set $\emptyset$, and each element is self-inverse. It remains to show associativity of the operation, so let $A$, $B$, and $C \in 2^S$. We show that $(A \triangle B) \triangle C = A \triangle (B \triangle C)$. Setting some notation, for a set $D \in 2^S$ we denote by $\mathbbm{1}_D \colon S \to \{0,1\}$ the characteristic function of $D$ such that, for $x \in S$, $x \mapsto 1$ if and only if $x \in D$. Further, denote by $\oplus$ the associative operation such that for any two characteristic functions of sets $D$, $E \in 2^S$ we set $(\mathbbm{1}_D \oplus \mathbbm{1}_E)(x) := \mathbbm{1}_D (x) + \mathbbm{1}_E (x) \pmod{2}$. Associativity of $\triangle$ then follows from
\begin{equation*}
\mathbbm{1}_{(A \triangle B) \triangle C} = \mathbbm{1}_{A \triangle B} \oplus \mathbbm{1}_C = (\mathbbm{1}_A \oplus \mathbbm{1}_B) \oplus \mathbbm{1}_C = \mathbbm{1}_A \oplus (\mathbbm{1}_B \oplus \mathbbm{1}_C) = \mathbbm{1}_A \oplus \mathbbm{1}_{B \triangle C} = \mathbbm{1}_{A \triangle (B \triangle C)}.\qedhere
\end{equation*}
\end{proof}

To a simplicial representation we now associate a group as in Proposition~\ref{prop:PowerSetGroup}, which we call a $G$-representation to indicate that the representation has the structure of a group. While we could construct the group without any reference to the simplicial representation, the group encodes the same model structure in terms of model components and their interconnections, so it is natural to indicate the relationship between the group and the corresponding simplicial representation.

\begin{definition}[\bf{$G$-representation}]\label{def:Grep}
Let $K$ be a simplicial representation of a model. The \emph{$G$-representation} of the model is the group $G_K := (2^K,\triangle)$.
\end{definition}

For a simplicial representation $K$, the elements of $G_K$ are subsets of simplicies from $K$. These subsets are of three types: a set containing a single simplex from $K$ is identified as a simplex; a set containing one or more simplices from $K$ is a simplicial subcomplex of $K$ if it is closed under taking faces, in particular a set containing just a vertex is both a simplex and a simplicial complex; a set containing two or more simplices from $K$ that is not a simplicial subcomplex of $K$ is a set of simplices.

\begin{notation}
Let $K$ be a simplicial representation of a model, with corresponding $G$-representation $G_K$. The subset $\Vertices (G_K) := \Big\{\, \big\{\{v\}\big\} \mid v \in \Vertices (K) \,\Big\}$ of $G_K$ is referred to as the vertices of $G_K$. Further, we denote by $\rho_K \colon 2^{\Vertices (G_K)} \setminus \emptyset \to K$ the map such that $\rho (S) = \big\{\cup_{i \in I} \{s_i\}\big\}$ for $S = \Big\{\big\{\{s_i\}\big\}\Big\}_{i \in I} \in 2^{\Vertices (G_K)} \setminus \emptyset$, which sends a set of vertices in $G_K$ to the simplex spanned by the corresponding vertices in $K$.
\end{notation}

Since simplicial representations are labelled simplicial complexes, the elements of $G$-representations are also labelled, and there is a canonical bijection $K \mapsto G_K$ from the set of simplicial representations of models onto the set of $G$-representations of models.

\begin{example}
Let $K$ be the 1-simplex with vertices labelled 1 and 2, that is, $K = \big\{ \{1\}, \{2\}, \{1,2\} \big\}$. Then $G_K = \Big\{ \emptyset, \big\{\{1\}\big\}, \big\{\{2\}\big\}, \big\{\{1,2\}\big\}, \big\{\{1\},\{2\}\big\}, \big\{\{1\},\{1,2\}\big\}, \big\{\{2\},\{1,2\}\big\}, \big\{\{1\},\{2\},\{1,2\}\big\}\Big\}$. Note that $\big\{\{1\}\big\}$ and $\big\{\{2\}\big\}$ correspond to both simplices and simplicial subcomplexes of $K$, $\big\{\{1,2\}\big\}$ corresponds to a simplex of $K$, $\big\{\{1\},\{2\}\big\}$ and $\big\{\{1\},\{2\},\{1,2\}\big\} = K$ correspond to simplicial subcomplexes of $K$, while each of the elements $\big\{\{1\},\{1,2\}\big\}$ and $\big\{\{2\},\{1,2\}\big\}$ are not closed under taking faces so are sets of simplices from $K$.
\end{example}

While the group $G_K$ contains subsets of $K$ that are not subcomplexes, we can describe $G_K$ in terms of the simplices in $K$. Recall that for a prime $p$, a $p$-group is a group in which every element has order a power of $p$. Since every nontrivial element of $G_K$ has order 2, $G_K$ is therefore an elementary commutative 2-group. Further, since $G_K$ is a finite $p$-group it follows from the Burnside Basis Theorem \cite[Chapter 5, Page 124, Theorem 5.50]{Rotman1995} that any two minimal generating sets of $G_K$ have equal cardinality, here equal to $\lvert K \rvert$, and are therefore generating sets of smallest cardinality.

\begin{proposition}\label{prop:Generators}
Let $K$ be a simplicial representation of a model. Then $\K := \big\{\, \{\sigma\} \mid \sigma \in K \,\big\}$ is a minimal set of generators for $G_K$.
\end{proposition}

\begin{proof}
Let $S \in G_K$ be nontrivial. Then $S \subseteq \K$ by the definition of $G_K$, so by associativity of the symmetric difference we have $S = \triangle_{\sigma \in S} \{\sigma\}$, and it follows that $\K$ generates $G_K$ as a group. Further, $\K$ is minimal since, for each $\{\sigma\} \in \K$, $\{\sigma\} \notin \big\langle \K \setminus \{\sigma\} \big\rangle \subseteq G_K$.
\end{proof}

Since $G_K$ is an elementary commutative 2-group it can be simply expressed as an internal direct sum of cyclic subgroups of subcomplexes of $K$, or isomorphically as an external direct sum of copies of $\ZZ / 2\ZZ$.

\begin{proposition}\label{prop:DirectSum}
Let $K$ be a simplicial representation of a model. Then $G_K \cong \bigoplus_{\sigma \in K} \big\langle \{\sigma\} \big\rangle \cong \bigoplus_{\sigma \in K} \ZZ / 2\ZZ$.
\end{proposition}

\begin{proof}
By the fundamental theorem of finite commutative groups (see \cite[Chapter 6, Page 128, Theorem 6.5]{Rotman1995}), $G_K$ is a direct sum of primary cyclic groups. Since $\K$ is a minimal set of generators for $G_K$ by Proposition~\ref{prop:Generators}, we therefore have $G_K \cong \bigoplus_{\sigma \in K} \big\langle \{\sigma\} \big\rangle$. Finally, each $\{\sigma\} \in \K$ has order 2 in $G_K$, so $\big\langle \{\sigma\} \big\rangle \cong \ZZ / 2\ZZ$.
\end{proof}

Viewing $G_K$ as the direct sum $\bigoplus_{\sigma \in K} \big\langle \{\sigma\} \big\rangle$ highlights the structure of the group, and allows for easier characterisation of, for example, homomorphisms between such groups.

Following from the definition of distance between two simplicial representations $K$ and $L$ of two models \cite{Vittadello2021b}, we can define the distance between the two $G$-representations $G_K$ and $G_L$ as the cardinality of the symmetric difference of $K$ and $L$. Indeed, this distance between $G$-representations can be formalised as a metric space, similar to the metric space of simplicial representations, and is also isometric to the latter.

The group $G_K$ associated with a simplicial representation $K$ is itself a representation of the model corresponding to $K$. The group elements of $G_K$ have labels that are induced by the labelled simplices contained in the group elements, which in turn have labels induced by the component labels of the model. The labels of the elements of $G_K$ must be accounted for by group homomorphisms when required. We therefore introduce our concept of a \emph{label-preserving group homomorphism}.

\begin{definition}[\textbf{Label-preserving group homomorphism}]
Let $G_K$ and $G_L$ be two $G$-representations associated with two simplicial representations $K$ and $L$ of two models. A group homomorphism $\phi \colon G_K \to G_L$ is \emph{label preserving} if whenever the group elements $g \in G_K$ and $h \in G_L$ have the same label then $\phi (g) = h$.
\end{definition}

Our definition of the equivalence of two $G$-representations of models is as follows.

\begin{definition}[\textbf{Equivalence of $G$-representations}]\label{def:EquivGRep}
Let $\C$ be the set of all components that appear in the collection of models under consideration for comparison, and let $\G$ be a collection of $G$-representations $G_K$ associated with simplicial representations which correspond to models generated by subsets from $\C$. Two $G$-representations $G_K$, $G_L \in \G$ are \emph{equivalent} if and only if there exist two (possibly empty) sequences of invertible partial operations $( f_i )_{i=0}^m$ and $( g_i )_{i=0}^n$ on $\G$ such that $f_m \circ \cdots \circ f_1 \circ f_0 (G_K) = g_n \circ \cdots \circ g_1 \circ g_0 (G_L)$.
\end{definition}

Note that the set $\{\, (G_K,G_L) \in \G \times \G \mid \text{$G_K$ and $G_L$ are equivalent} \,\}$ is an equivalence relation on $\G$. For the equivalence of $G$-representations to be conceptually meaningful, just as for the equivalence of simplicial representations, we need to ensure that the employed partial operations are themselves conceptually meaningful. We consider specific operations on $G$-representations that are analogous to those for simplicial representations. For the $G$-representations, however, the operations are now induced by label-preserving group homomorphisms.

For brevity we only outline the operations on $G$-representations, as they correspond closely with the operations for simplicial operations. Note that for $G$-representations the action of an adjacent/nonadjacent-vertex identification does not send two vertices to the same vertex as for simplicial representations, but rather removes all group elements associated with one of the two vertices. This action is essentially the same as that for simplicial representations, and allows for a description in terms of group homomorphisms. Further, we omit definitions for the vertex split operation and the inclusion operation, both of which can be defined in terms of canonical injections between the direct-sum decompositions of appropriate $G$-representations, since we explicitly describe the equivalence of $G$-representations with the operations of adjacent/nonadjacent-vertex identifications and vertex substitution, while the vertex split and inclusion operations are employed implicitly. Similar to the operations for simplicial representations, all five of these partial operations on a set of $G$-representations are invertible for suitable domains of definition, and the equivalence of $G$-representations must be based on the admissibility of the partial operations.

For all operations, let $\C$ be the set of all components of models under consideration, and let $G_K$ and $G_L$ be $G$-representations associated with the two simplicial representations $K$ and $L$ with labels from $\C$.

\begin{description}
\item[Adjacent-vertex identification:] Suppose $\big\{ \{u, v\} \big\} \in G_K$, where $u$, $v \in \Vertices (K)$, and that assumptions hold for $u$ and $v$ corresponding to those in Definition~\ref{def:Op1}. A label-preserving group homomorphism $\phi_1 \colon G_K \to G_L$ is an \emph{adjacent-vertex identification} if $G_L = \bigoplus_{ \{ \sigma \in K \mid v \notin \sigma \}} \big\langle \{\sigma\} \big\rangle$ and $\phi_1 \colon \bigoplus_{\sigma \in K} \big\langle \{\sigma\} \big\rangle \to \bigoplus_{ \{ \sigma \in K \mid v \notin \sigma \} } \big\langle \{\sigma\} \big\rangle$ is the canonical projection.
\item[Nonadjacent-vertex identification:] Suppose $\big\{ \{u\} \big\}$, $\big\{ \{v\} \big\} \in G_K$, $\big\{ \{u, v\} \big\} \notin G_K$, and that assumptions hold for $u$ and $v$ corresponding to those in Definition~\ref{def:Op2}. A label-preserving group homomorphism $\phi_2 \colon G_K \to G_L$ is a \emph{nonadjacent-vertex identification} if $G_L = \bigoplus_{ \{ \sigma \in K \mid v \notin \sigma \}} \big\langle \{\sigma\} \big\rangle$ and $\phi_2 \colon \bigoplus_{\sigma \in K} \big\langle \{\sigma\} \big\rangle \to \bigoplus_{ \{ \sigma \in K \mid v \notin \sigma \} } \big\langle \{\sigma\} \big\rangle$ is the canonical projection.
\item[Vertex substitution:] Suppose $u \in \Vertices (K)$, $v \in \Vertices (L)$, and $\Vertices (K) \setminus \{u\} = \Vertices (L) \setminus \{v\}$. Let $\pi \colon \Vertices (K) \to \Vertices (L)$ be a map that is injective except that $\pi (u) = v$. A label-preserving group isomorphism $\phi_5 \colon G_K \to G_L$ is a \emph{vertex substitution} if, for each $\sigma \in K$, $\phi_5 (\{\sigma\}) = \big\{\{\, \pi (w) \mid w \in \sigma \,\}\big\}$.
\end{description}

To determine whether there are conceptually equivalent vertices in a $G$-representation $G_K$, we can employ group actions on $G_K$ that correspond to exchange automorphisms for simplicial representations. Recall that an action of a group $H$ on a group $G$ is a group homomorphism $\Theta \colon H \to \Aut (G)$. While we required orbit spaces for vertex identifications when considering simplicial representations, for $G$-representations we have no need for orbit spaces and can employ $G$-representations directly. The methodology for determining model equivalence detailed in Subsection~\ref{subsec:method} is easily adapted for $G$-representations.

\section{Conclusion} \label{sec:Conclusion}
In this article we develop a rigorous and automatable methodology for determining whether vertices in a simplicial representation are conceptually related, corresponding to conceptually-related model components, and then identifying these vertices using simplicial operations. This process is the main consideration of comparison by equivalence as defined in Vittadello and Stumpf \cite{Vittadello2021b}, and exploits the symmetry associated with conceptually equivalent vertices in a simplicial representation by constructing group actions on the simplicial representations. We develop the required mathematical theory for group actions on simplicial complexes, and demonstrate how this methodology greatly simplifies and facilitates the process of determining model equivalence by reducing the initial simplicial representations to simplicial complexes with no conceptually equivalent vertices. Our approach to model comparison by equivalence provides a rigorous and efficient framework for comparing the underlying conceptual structure of mathematical models of a particular biological system, as well as mathematical models of different biological systems.

The important point to note is that our approach, like our previous one, differs from the overwhelming majority of model comparison approaches by being independent of outputs: it is solely based on the underlying conceptual structure of the model. Our definition of model structure is also more inclusive than, for example, motif-based definitions\cite{Alon2007}, as these are known to fall short of constraining dynamics \cite{Ingram2006}. By including model properties such as boundary conditions and symmetry relationships, for example, as components of the simplicial complexes, our approach here is capable of capturing even subtle differences in the dynamics. Our methodology is applicable to all models irrespective of mathematical formalism, and allows for comparison from different conceptual perspectives.

Model comparison is essential for managing the seemingly heterogeneous and possibly numerous models associated with a system. We conclude this analysis with a discussion of an alternative mathematical framework for our model comparison methodology by representing models as groups, which provides for the direct application of techniques from group theory within our model-comparison methodology. That we can represent our model comparison methodology with two distinct mathematical frameworks shows the flexibility of our approach, which is important for the utilisation of a broad range of mathematical techniques when working on the difficult problem of model comparison.

Modelling in biology cannot rely on the same fundamental principles --- such as symmetries and conservation laws --- that have been so successful in physics. Instead, statistical methods \cite{Kirk2013}, robustness analysis \cite{Bates2011}, or algebraic methods \cite{Araujo2018} are required to narrow down the vast `universe' of potential mathematical models in order to identify those models that can reasonably explain a given scenario. The approach developed here is capable of taking these models, including large sets of good candidates \cite{Scholes2019}, and distilling design principles from this analysis. Design principles, or the conceptual bases for mathematical models describing a given system, are arguably our best chance to gain conceptual mechanistic insights into biological systems.

\section{Appendix} \label{sec:Appendix}
Here we provide the details of the mathematical models that we use in the example in Subsection~\ref{subsec:example}. We also provide a brief discussion on representing the models as simplicial complexes.

These models are from the two main classes of models for developmental pattern formation, namely Turing-pattern (TP) models and positional-information (PI) models. We consider the patterning to occur within a two-dimensional rectangular domain, with zero-flux boundary conditions on the two sides parallel to the morphogen concentration gradient. We assume that the velocities of the cytoplasm and the growing tissue are negligible; we therefore assume no advection. Further details of the models and the associated simplicial representations, including the vertex numbers associated with the simplicial representations, are in Vittadello and Stumpf~\cite{Vittadello2021b}. Note that the boundary conditions for the TP models may be Dirichlet, Neumann, Robin, or periodic \cite{Dillon1994,Varea1997,Barrio1999}.

\subsection{Positional-information: Linear gradient}
A linear concentration profile of a morphogen results when the production and degradation of the morphogen occur outside and on opposite sides of the tissue domain, and the morphogen undergoes passive diffusion along the domain from the side where it is produced to the side where it is degraded \cite{Stumpf1967,Wolpert1969,Crick1970,Capek2019}. Mathematically, the steady-state morphogen concentration in this system satisfies Laplace’s equation, which yields global scale-invariant positional information. Specifically, if the tissue length is $L$ with initial position at $x = 0$ then the morphogen concentration $m(x,t)$ can be modelled as
\begin{equation}\label{eq:PI1}
\frac{\partial m}{\partial t} = D \frac{\partial^2 m}{\partial x^2},
\end{equation}
where $D$ is the morphogen diffusivity. The steady-state solution of Equation~\ref{eq:PI1} is the linear equation $m(x) = ax+b$, for arbitrary constants $a$, $b \in \RR$. The boundary conditions are influx at one end and outflux at the other end, and we may assume without loss of generality that influx occurs at $x = 0$ and outflux at $x = L$. The original Dirichlet boundary conditions specify the constant concentrations $m(0,t) = m_0 > 0$ and $m(L,t) = m_L \ge 0$, with $m_0 > m_L$, so the solution is $m(x) = \big((m_L - m_0)/L\big) x + m_0$. Monotonically decreasing linear gradients can also be achieved with Neumann boundary conditions at one boundary and Dirichlet boundary conditions at the opposite boundary.

The simplicial representation of the linear model is shown in Figure~\ref{fig:Fig2}(a). Note that the simplicial complex is 4-dimensional, however we only show the 1-skeleton of the simplicial complex for simplicity.

\subsection{Positional-information: Synthesis-diffusion-degradation (SDD)}
In the SDD model a morphogen gradient forms by morphogen production from a localised source at the boundary combined with morphogen diffusion and uniform degradation throughout the tissue \cite{Wartlick2009,Umulis2013,Capek2019}. Mathematically, the morphogen concentration $m(\bm{x},t)$ can be modelled as
\begin{equation}\label{eq:PI2}
\frac{\partial m}{\partial t} = D \nabla^2 m - km,
\end{equation}
where $D$ is the morphogen diffusivity, and $k$ is the morphogen degradation rate which is independent of morphogen concentration. The boundary conditions are Neumann for the influx and zero outflux. The steady-state morphogen gradient is an exponential function, so this system is not scale invariant unless $D$, $k$, and the influx vary with the tissue length $L$ in a very specific manner \cite{Umulis2013}.

The simplicial representation of the SDD model is shown in Figure~\ref{fig:Fig2}(b). Note that the simplicial complex is 3-dimensional, however we only show the 1-skeleton of the simplicial complex for simplicity.

\subsection{Positional-information: Opposing gradients}
There are two mechanisms whereby two opposing morphogen gradients provide size information for developmental patterning \cite{McHale2006}: the gene expression depends on the relative concentrations of the two morphogens. One is the annihilation model, and the other is the scaling-by-opposing-gradients model.

For the opposing gradients model the sources of each morphogen are at opposite ends of the domain. The combination of the two morphogen gradients provides effective local scaling for the boundaries of individual target genes, however it does not produce global scaling across the entire domain \cite{Wolpert1969,Houchmandzadeh2005,McHale2006,Capek2019}. The degree of scaling for the opposing gradients mechanism can be increased if the morphogens are irreversibly inactivated upon binding to each other \cite{McHale2006}. Here we assume no direct interaction between the two morphogens. Mathematically, the two morphogen gradients with concentrations $m(\bm{x},t)$ and $c(\bm{x},t)$ can be modelled as
\begin{align}\label{eq:PI3}
\frac{\partial m}{\partial t} &= D_m \nabla^2 m - k_m m,\\
\frac{\partial c}{\partial t} &= D_c \nabla^2 c - k_c c,
\end{align}
where $D_m$ and $D_c$ are diffusivities, and $k_m$ and $k_c$ are degradation rates. The boundary conditions for each morphogen are Neumann at both boundaries, with flux at the source and zero flux at the opposite boundary.

The simplicial representation of the opposing gradients model is shown in Figure~\ref{fig:Fig2}(c). Note that the simplicial complex is 4-dimensional, however we only show the 1-skeleton of the simplicial complex for simplicity.

\subsection{Positional-information: Induction-contraction}
The largest degree of scaling across the patterning domain is obtained when the biophysical properties of the morphogen are influenced by an accessory modulator that senses the size of the domain \cite{Wartlick2009,Ben_Zvi2010,Umulis2013,Capek2019}. The modulator molecules with concentration $c(\bm{x},t)$ may influence the diffusivity, degradation rate, or influx rate, of the morphogen with concentration $m(\bm{x},t)$. If the kinetics and source of the modulator are dependent on the morphogen concentration then the mechanism is \emph{active modulation}. An example of active modulation is the \emph{induction-contraction} model, which scales globally \cite{Rahimi2016,Shilo2017}. In this case, the morphogen induces the production of a contractor molecule, here the modulator with high diffusivity once again, which contracts the range of the morphogen gradient through a decrease in the morphogen diffusivity or an increase of the morphogen degradation rate. Note that while the amplitude and shape of the morphogen gradient is globally scale-invariant, the same is not necessarily true of the modulator since the modulator level reflects the domain size and therefore increases or decreases accordingly.

The scaling by modulation model with the induction-contraction mechanism, which can be represented mathematically as
\begin{align}\label{eq:PI5}
\frac{\partial m}{\partial t} &= D_m (c) \nabla^2 m - k_m (c) m,\\
\frac{\partial c}{\partial t} &= D_c \nabla^2 c - k_c c + \rho (m),
\end{align}
where $D_m (c)$ and $D_c$ are diffusivities, $k_m (c)$ and $k_c$ are degradation rates, and $\rho (m)$ is the localised production source for the modulator. Note that $D_m (c)$ is a decreasing function of $c$, $k_m (c)$ is an increasing function of $c$, and $\rho (m)$ is an increasing function of $m$. Further note that `degradation' in this context refers to not only physical destruction, but to any mechanism that effects the removal of the morphogen from the patterning system, including irreversible complex formation. Morphogen boundary conditions are Neumann for the influx and zero flux at the opposite end of the domain. For the modulator there are various possibilities for the boundary conditions, such as the same type of boundary conditions as the morphogen when the modulator source is outside the domain or, as we consider here, zero flux at both boundaries when the modulator source is within the domain.

\subsection{Positional-information: annihilation model}
The annihilation model, and also the opposing gradients model are mechanisms whereby two opposing morphogen gradients provide size information for developmental patterning. In the annihilation model, the target gene responds to the concentration of Morphogen 1, to which Morphogen 2 irreversibly binds and thereby inhibits the action of Morphogen 1 on activity of transcription, so that the gradient of Morphogen 2 provides size information to the concentration field of Morphogen 1 \cite{McHale2006}.

The sources of each morphogen are at opposite ends of the domain, and the morphogens interact by an annihilation reaction with rate $k$ that results in global scale-invariant patterning \cite{Ben_Naim1992,McHale2006}. Mathematically, the two morphogen gradients with concentrations $m(\bf{x},t)$ and $c(\bf{x},t)$ can be modelled as
\begin{align}\label{eq:PI4}
\frac{\partial m}{\partial t} &= D_m \nabla^2 m - k_m m - kmc,\\
\frac{\partial c}{\partial t} &= D_c \nabla^2 c - k_c c - kmc,
\end{align}
where $D_m$ and $D_c$ are diffusivities, and $k_m$ and $k_c$ are degradation rates. The boundary conditions for each morphogen are Neumann at both boundaries, with influx at the source and zero flux at the opposite boundary.

The simplicial representation of the annihilation mode is shown in Figure~\ref{fig:Fig3}(a). Note that the simplicial complex is 4-dimensional, however we only show the 1-skeleton of the simplicial complex for simplicity.

\subsection{Turing-pattern: Activator-inhibitor model}
The activator-inhibitor system \cite{Gierer1972,Meinhardt2012,Landge2020} consists of two diffusible morphogens, an autocatalytic activator with concentration $m(\bf{x},t)$ and a rapidly diffusing inhibitor with concentration $c(\bf{x},t)$. This model can be represented mathematically as
\begin{align}\label{eq:T1}
\frac{\partial m}{\partial t} &= D_m \nabla^2 m + \frac{\rho m^2}{c(1+\mu_m m^2)} - k_m m + \rho_m,\\
\frac{\partial c}{\partial t} &= D_c \nabla^2 c + \rho m^2 - k_c c + \rho_c,
\end{align}
where $D_m$ and $D_c$ are diffusivities, $\rho_m$ and $\rho_c$ are basal production rates, $k_m$ and $k_c$ are degradation rates, and $\mu_m$ is a saturation constant. The parameter $\rho$ is the \emph{source density}, which measures the general ability of the cells to perform the autocatalytic reaction. The patterning arises through local self-enhancement of the activator, activation of the inhibitor, and long-range inhibition of the activator. We assume that the boundary conditions are zero flux at both boundaries.

The simplicial representation of the activator-inhibitor model is shown in Figure~\ref{fig:Fig3}(b). Note that the simplicial complex is 5-dimensional, however we only show the 1-skeleton of the simplicial complex for simplicity.

\subsection{Turing-pattern: Substrate depletion}
While long-range inhibition of an autocatalytic activator can occur with an inhibitor, as in the activator-inhibitor model, the antagonistic effect can also arise from the depletion of a substrate that is consumed during activator production. In this system, the production of an autocatalytic activator with concentration $m(\bm{x},t)$ results in either direct or indirect depletion of a substrate with concentration $s(\bm{x},t)$ \cite{Meinhardt2012}. This model can be represented mathematically as
\begin{align}\label{eq:T2}
\frac{\partial m}{\partial t} &= D_m \nabla^2 m + \rho s m^2 - k_m m + \rho_m,\\
\frac{\partial s}{\partial t} &= D_s \nabla^2 s - \rho s m^2 - k_s s + \rho_s,
\end{align}
where $D_m$ and $D_s$ are diffusivities, $\rho_m$ and $\rho_s$ are basal production rates, $k_m$ and $k_s$ are degradation rates, and $\rho$ is the source density for the autocatalytic reaction of the activator, similar to the activator-inhibitor model. The diffusivity $D_s$ of the substrate must be much faster than the diffusivity $D_m$ of the activator, and it is assumed that the substrate is produced uniformly throughout the domain. We assume that the boundary conditions are zero flux at both boundaries.

The simplicial representation of the substrate depletion model is shown in Figure~\ref{fig:Fig3}(c). Note that the simplicial complex is 5-dimensional, however we only
show the 1-skeleton of the simplicial complex for simplicity.

\subsection{Turing-pattern: Inhibition of an inhibition}
Dynamics analogous to the activator-inhibitor model can be realised through an inhibition of an inhibition mechanism \cite{Meinhardt2012}. In this case, two morphogens with concentrations $a(\bm{x},t)$ and $c(\bm{x},t)$ inhibit the production of each other, thereby forming a switching system in which one of the morphogens becomes fully activated similar to being autcatalytic. In order that pattern formation occur, a third morphogen with concentration $b(\bm{x},t)$ acts as a long-range signal that disrupts the indirect self-enhancement of either $a$ or $c$. Morphogen $b$ is rapidly diffusing, is produced under control of $a$, and inhibits the inhibition of $c$ production by $a$. therefore acting as inhibitor. This model can be represented mathematically as
\begin{align}\label{eq:T3}
\frac{\partial a}{\partial t} &= D_a \nabla^2 a + \frac{\rho_a}{\kappa_a + c^2} - k_a a,\\
\frac{\partial b}{\partial t} &= D_b \nabla^2 b + k_b (a-b),\\
\frac{\partial c}{\partial t} &= D_c \nabla^2 c + \frac{\rho_c}{\kappa_c + a^2/b^2} - k_c c,
\end{align}
where $D_a$, $D_b$, and $D_c$ are diffusivities, $\rho_a$ and $\rho_c$ are production rates, $k_a$ and $k_c$ are degradation rates, $k_b$ is both the rate of production and degradation of $b$, and $\kappa_a$ and $\kappa_c$ are saturation constants that limit the production rate if the concentrations of $a$ or $c$ become too low. We assume that the boundary conditions are zero flux at both boundaries.

\subsection{Turing-pattern: Modulation}
The diffusion of a morphogen may be inhibited by adsorption on negatively-charged extracellular matrix (ECM) components, resulting in modulated diffusion and a smaller effective diffusivity for the morphogen \cite{Nesterenko2017}. The corresponding modulation model \cite{Nesterenko2017} is an extension of the activator-inhibitor model \cite{Gierer1972,Meinhardt2012,Landge2020}. Our description of the modulation model is based on the version of the activator-inhibitor model described above. The modulation model consists of two diffusible morphogens, an autocatalytic activator and an inhibitor, along with available binding sites on the ECM onto which the activator adsorbs. The inhibitor does not bind to the ECM, and the free activator and the inhibitor have equal diffusivities $D$. It is assumed that autocatalysis and activation of the inhibitor by the activator occurs in both the free and adsorbed states. It is also assumed that the degradation rate of the activator is equal in both the free and adsorbed states. Free binding sites on the ECM therefore appear due to both desorption and degradation of the activator. Such modulation allows for stable dissipative morphogens gradients. Mathematically, the free activator has concentration $a (\bm{x},t)$, the bound activator has concentration $b (\bm{x},t)$, the inhibitor has concentration $c(\bm{x},t)$, and the concentration of available binding sites on the ECM is $s(\bm{x},t)$:
\begin{align}\label{eq:T4}
\frac{\partial a}{\partial t} &= D\nabla^2 a + \frac{\rho (a + b)^2}{c(1 + \mu_a (a + b)^2)} - k_a a + \rho_{a} - k_1 sa + k_{-1} b,\\
\frac{\partial c}{\partial t} &= D\nabla^2 c + \rho (a + b)^2 - k_c c + \rho_c,\\
\frac{\partial s}{\partial t} &= - k_1 sa + (k_{-1} + k_a)b,
\end{align}
where $D$ is the diffusivity, $\rho_{a}$ and $\rho_c$ are basal production rates, $k_a$ and $k_c$ are degradation rates, $\mu_a$ is the saturation constant, $\rho$ is the source density, $k_1$ is the rate of adsorption of the activator onto the ECM, and $k_{-1}$ is the rate of desorption of the activator from the ECM. We assume that the boundary conditions are zero flux at both boundaries.

\subsection{Construction of the simplicial representations}
The simplicial representations of the four Turing-pattern and five positional-information models are constructed by first determining the set of components $\C$ on which all of the the models under consideration are based. In this case the general components are: agents, namely morphogens, modulators, and substrates; reactions involving the agents, such as self-activation, activation, inhibition and annihilation; agent degradation; influx and outflux boundary conditions; agent diffusion; profile of the morphogen gradient; and scale invariance of the morphogen gradient. The ordered set of these components is shown in Figure 3 in \cite{Vittadello2021b}.

The reference complex is the simplex spanned by the complete set of components for all nine models. To construct the simplicial representation of each Turing-pattern and positional-information model, we first specify the 0-simplices, which represent the model components, and the 1-simplices, which represent direct interconnections between the components. While the 0- and 1-simplices are specified by the model, to give a combinatorial graph, for these models we obtain higher-dimensional simplices by forming cliques where possible, incrementally in dimensions 2 and higher. These higher-dimensional simplices indicate higher-dimensional interactions between the corresponding model components.

\section*{Author contributions}
S.T.V. conceived, planned, and conducted this research; S.T.V. drafted the manuscript; S.T.V. and M.P.H.S. reviewed, edited, and approved the final version.

\section*{Funding}
The authors gratefully acknowledge funding through a ``Life?'' programme grant from the Volkswagen Stiftung. M.P.H.S. is funded through the University of Melbourne Driving Research Momentum program.

\section*{Competing interests}
We declare we have no competing interests.

\section*{Acknowledgments}
We thank members of the Theoretical Systems Biology group at the University of Melbourne and Imperial College London, Heike Siebert (FU Berlin), James Briscoe (The Francis Crick Institute) and Mark Isalan (Imperial College London) for helpful discussions on Turing patterns and model discrimination.

\clearpage

\newcommand{\noop}[1]{}


\begin{thebibliography}{10}

\bibitem{Howe2008}
Howe D, Costanzo M, Fey P, Gojobori T, Hannick L, Hide W, et~al.
\newblock The future of biocuration.
\newblock \emph{Nature}. 2008;455:47--50.
\newblock doi:10.1038/455047a.

\bibitem{Marx2013}
Marx V.
\newblock The big challenges of big data.
\newblock \emph{Nature}. 2013;498:255--260.
\newblock doi:10.1038/498255a.

\bibitem{ISB2018}
{International Society for Biocuration}.
\newblock Biocuration: distilling data into knowledge.
\newblock \emph{{PLOS} Biology}. 2018;16:e2002846.
\newblock doi:10.1371/journal.pbio.2002846.

\bibitem{Mahmud2021}
Mahmud M, Kaiser MS, McGinnity TM, Hussain A.
\newblock Deep learning in mining biological data.
\newblock \emph{Cognitive Computation}. 2021;13:1--33.
\newblock doi:10.1007/s12559-020-09773-x.

\bibitem{Tomlin2007}
Tomlin CJ, Axelrod JD.
\newblock Biology by numbers: mathematical modelling in developmental biology.
\newblock \emph{Nature Reviews Genetics}. 2007;8:331--340.
\newblock doi:10.1038/nrg2098.

\bibitem{Sneddon2010}
Sneddon MW, Faeder JR, Emonet T.
\newblock Efficient modeling, simulation and coarse-graining of biological
  complexity with {NFsim}.
\newblock \emph{Nature Methods}. 2010;8(2):177--183.
\newblock doi:10.1038/nmeth.1546.

\bibitem{Gunawardena2014}
Gunawardena J.
\newblock Models in biology: `accurate descriptions of our pathetic thinking'.
\newblock \emph{{BMC} Biology}. 2014;12:29.
\newblock doi:10.1186/1741-7007-12-29.

\bibitem{Wolkenhauer2014}
Wolkenhauer O.
\newblock Why model?
\newblock \emph{Frontiers in Physiology}. 2014;5:21.
\newblock doi:10.3389/fphys.2014.00021.

\bibitem{Torres2015}
Torres NV, Santos G.
\newblock The (mathematical) modeling process in biosciences.
\newblock \emph{Frontiers in Genetics}. 2015;6:354.
\newblock doi:10.3389/fgene.2015.00354.

\bibitem{Pezzulo2016}
Pezzulo G, Levin M.
\newblock Top-down models in biology: explanation and control of complex living
  systems above the molecular level.
\newblock \emph{Journal of The Royal Society Interface}. 2016;13:20160555.
\newblock doi:10.1098/rsif.2016.0555.

\bibitem{Transtrum2016}
Transtrum MK, Qiu P.
\newblock Bridging mechanistic and phenomenological models of complex
  biological systems.
\newblock \emph{{PLOS} Computational Biology}. 2016;12:e1004915.
\newblock doi:10.1371/journal.pcbi.1004915.

\bibitem{Banwarth_Kuhn2020}
Banwarth-Kuhn M, Sindi S.
\newblock How and why to build a mathematical model: a case study using prion
  aggregation.
\newblock \emph{Journal of Biological Chemistry}. 2020;295:5022--5035.
\newblock doi:10.1074/jbc.REV119.009851.

\bibitem{King2021}
King J, Eroum{\'{e}} KS, Truckenm{\"{u}}ller R, Giselbrecht S, Cowan AE, Loew
  L, et~al.
\newblock Ten steps to investigate a cellular system with mathematical
  modeling.
\newblock \emph{{PLOS} Computational Biology}. 2021;17:e1008921.
\newblock doi:10.1371/journal.pcbi.1008921.

\bibitem{Rosenblueth1945}
Rosenblueth A, Wiener N.
\newblock The role of models in science.
\newblock \emph{Philosophy of Science}. 1945;12:316--321.
\newblock doi:10.1086/286874.

\bibitem{Karplus1977}
Karplus WJ.
\newblock The spectrum of mathematical modeling and systems simulation.
\newblock \emph{Mathematics and Computers in Simulation}. 1977;19:3--10.
\newblock doi:10.1016/0378-4754(77)90034-9.

\bibitem{Babtie2014}
Babtie AC, Kirk P, Stumpf MPH.
\newblock Topological sensitivity analysis for systems biology.
\newblock \emph{Proceedings of the National Academy of Sciences}.
  2014;111:18507--18512.
\newblock doi:10.1073/pnas.1414026112.

\bibitem{Gay2010}
Gay S, Soliman S, Fages F.
\newblock A graphical method for reducing and relating models in systems
  biology.
\newblock \emph{Bioinformatics}. 2010;26:i575--i581.
\newblock doi:10.1093/bioinformatics/btq388.

\bibitem{Henkel2018}
Henkel R, Hoehndorf R, Kacprowski T, Kn{\"{u}}pfer C, Liebermeister W,
  Waltemath D.
\newblock Notions of similarity for systems biology models.
\newblock \emph{Briefings in Bioinformatics}. 2018;19:77--88.
\newblock doi:10.1093/bib/bbw090.

\bibitem{Tapinos2013}
Tapinos A, Mendes P.
\newblock A method for comparing multivariate time series with different
  dimensions.
\newblock \emph{{PLoS ONE}}. 2013;8:e54201.
\newblock doi:10.1371/journal.pone.0054201.

\bibitem{Cabbia2020}
Cabbia A, Hilbers PAJ, {van Riel} NAW.
\newblock A distance-based framework for the characterization of metabolic
  heterogeneity in large sets of genome-scale metabolic models.
\newblock \emph{Patterns}. 2020;1:100080.
\newblock doi:10.1016/j.patter.2020.100080.

\bibitem{Kirk2013}
Kirk P, Thorne T, Stumpf MPH.
\newblock Model selection in systems and synthetic biology.
\newblock \emph{Current Opinion in Biotechnology}. 2013;24:767--774.
\newblock doi:10.1016/j.copbio.2013.03.012.

\bibitem{Barnes2011}
Barnes CP, Silk D, Sheng X, Stumpf MPH.
\newblock Bayesian design of synthetic biological systems.
\newblock \emph{Proceedings of the National Academy of Sciences}.
  2011;108:15190--15195.
\newblock doi:10.1073/pnas.1017972108.

\bibitem{Brophy2014}
Brophy JAN, Voigt CA.
\newblock Principles of genetic circuit design.
\newblock \emph{Nature Methods}. 2014;11:508--520.
\newblock doi:10.1038/nmeth.2926.

\bibitem{Vittadello2021b}
Vittadello ST, Stumpf MPH.
\newblock Model comparison via simplicial complexes and persistent homology.
\newblock \emph{Royal Society Open Science}. 2021;8:211361.
\newblock doi:10.1098/rsos.211361.

\bibitem{Green2015}
Green JBA, Sharpe J.
\newblock Positional information and reaction-diffusion: two big ideas in
  developmental biology combine.
\newblock \emph{Development}. 2015;142:1203--1211.
\newblock doi:10.1242/dev.114991.

\bibitem{Scholes2019}
Scholes NS, Schnoerr D, Isalan M, Stumpf MPH.
\newblock A comprehensive network atlas reveals that {T}uring patterns are
  common but not robust.
\newblock \emph{Cell Systems}. 2019;9:243--257.
\newblock doi:10.1016/j.cels.2019.07.007.

\bibitem{Spanier1966}
Spanier EH.
\newblock Algebraic Topology.
\newblock Springer New York; 1966.
\newblock doi:10.1007/978-1-4684-9322-1.

\bibitem{Rotman1988}
Rotman JJ.
\newblock An Introduction to Algebraic Topology.
\newblock Springer New York; 1988.
\newblock doi:10.1007/978-1-4612-4576-6.

\bibitem{Munkres2018}
Munkres J.
\newblock Elements of Algebraic Topology.
\newblock {CRC} Press; 2018.
\newblock doi:10.1201/9780429493911.

\bibitem{Kozlov2008}
Kozlov D.
\newblock Combinatorial Algebraic Topology. vol.~21 of Algorithms and
  Computation in Mathematics.
\newblock Springer Berlin Heidelberg; 2008.
\newblock doi:10.1007/978-3-540-71962-5.

\bibitem{Bredon1972}
Bredon GE.
\newblock Introduction to Compact Transformation Groups. vol.~46 of Pure and
  Applied Mathematics.
\newblock New York: Academic Press; 1972.

\bibitem{Papadopoulos2012}
Papadopoulos A, editor.
\newblock Handbook of Teichm{\"{u}}ller theory. vol. Volume III.
\newblock European Mathematical Society; 2012.

\bibitem{Rotman1995}
Rotman JJ.
\newblock An Introduction to the Theory of Groups.
\newblock Springer New York; 1995.
\newblock doi:10.1007/978-1-4612-4176-8.

\bibitem{Alon2007}
Alon U.
\newblock Network motifs: theory and experimental approaches.
\newblock \emph{Nature Reviews Genetics}. 2007;8:450--461.
\newblock doi:10.1038/nrg2102.

\bibitem{Ingram2006}
Ingram PJ, Stumpf MPH, Stark J.
\newblock Network motifs: structure does not determine function.
\newblock \emph{{BMC} Genomics}. 2006 may;7(1):108.
\newblock doi:10.1186/1471-2164-7-108.

\bibitem{Bates2011}
Bates DG, Cosentino C.
\newblock Validation and invalidation of systems biology models using
  robustness analysis.
\newblock \emph{{IET} Systems Biology}. 2011;5:229--244.
\newblock doi:10.1049/iet-syb.2010.0072.

\bibitem{Araujo2018}
Araujo RP, Liotta LA.
\newblock The topological requirements for robust perfect adaptation in
  networks of any size.
\newblock \emph{Nature Communications}. 2018;9:1757.
\newblock doi:10.1038/s41467-018-04151-6.

\bibitem{Dillon1994}
Dillon R, Maini PK, Othmer HG.
\newblock Pattern formation in generalized {T}uring systems.
\newblock \emph{Journal of Mathematical Biology}. 1994;32:345--393.
\newblock doi:10.1007/BF00160165.

\bibitem{Varea1997}
Varea C, Arag{\'{o}}n JL, Barrio RA.
\newblock Confined {T}uring patterns in growing systems.
\newblock \emph{Physical Review E}. 1997;56:1250--1253.
\newblock doi:10.1103/PhysRevE.56.1250.

\bibitem{Barrio1999}
Barrio RA, Varea C, Arag{\'{o}}n JL, Maini PK.
\newblock A two-dimensional numerical study of spatial pattern formation in
  interacting {T}uring systems.
\newblock \emph{Bulletin of Mathematical Biology}. 1999;61:483--505.
\newblock doi:10.1006/bulm.1998.0093.

\bibitem{Stumpf1967}
Stumpf HF.
\newblock {\"{U}}ber den {V}erlauf eines schuppenorientierenden
  {G}ef{\"{a}}lles bei \emph{Galleria mellonella}.
\newblock \emph{Wilhelm Roux{\textquotesingle} Archiv f{\"{u}}r
  Entwicklungsmechanik der Organismen}. 1967;158:315--330.
\newblock doi:10.1007/bf00573402.

\bibitem{Wolpert1969}
Wolpert L.
\newblock Positional information and the spatial pattern of cellular
  differentiation.
\newblock \emph{Journal of Theoretical Biology}. 1969;25:1--47.
\newblock doi:10.1016/S0022-5193(69)80016-0.

\bibitem{Crick1970}
Crick F.
\newblock Diffusion in embryogenesis.
\newblock \emph{Nature}. 1970;225:420--422.
\newblock doi:10.1038/225420a0.

\bibitem{Capek2019}
{\v{C}}apek D, M{\"{u}}ller P.
\newblock Positional information and tissue scaling during development and
  regeneration.
\newblock \emph{Development}. 2019;146:dev177709.
\newblock doi:10.1242/dev.177709.

\bibitem{Wartlick2009}
Wartlick O, Kicheva A, Gonz{\'{a}}lez-Gait{\'{a}}n M.
\newblock Morphogen gradient formation.
\newblock \emph{Cold Spring Harbor Perspectives in Biology}. 2009;1:a001255.
\newblock doi:10.1101/cshperspect.a001255.

\bibitem{Umulis2013}
Umulis DM, Othmer HG.
\newblock Mechanisms of scaling in pattern formation.
\newblock \emph{Development}. 2013;140:4830--4843.
\newblock doi:10.1242/dev.100511.

\bibitem{McHale2006}
McHale P, Rappel WJ, Levine H.
\newblock Embryonic pattern scaling achieved by oppositely directed morphogen
  gradients.
\newblock \emph{Physical Biology}. 2006;3:107--120.
\newblock doi:10.1088/1478-3975/3/2/003.

\bibitem{Houchmandzadeh2005}
Houchmandzadeh B, Wieschaus E, Leibler S.
\newblock Precise domain specification in the developing \emph{Drosophila}
  embryo.
\newblock \emph{Physical Review E}. 2005;72:061920.
\newblock doi:10.1103/PhysRevE.72.061920.

\bibitem{Ben_Zvi2010}
Ben-Zvi D, Barkai N.
\newblock Scaling of morphogen gradients by an expansion-repression integral
  feedback control.
\newblock \emph{Proceedings of the National Academy of Sciences}.
  2010;107:6924--6929.
\newblock doi:10.1073/pnas.0912734107.

\bibitem{Rahimi2016}
Rahimi N, Averbukh I, Haskel-Ittah M, Degani N, Schejter ED, Barkai N, et~al.
\newblock A {WntD}-dependent integral feedback loop attenuates variability in
  \emph{Drosophila} toll signaling.
\newblock \emph{Developmental Cell}. 2016;36:401--414.
\newblock doi:10.1016/j.devcel.2016.01.023.

\bibitem{Shilo2017}
Shilo BZ, Barkai N.
\newblock Buffering global variability of morphogen gradients.
\newblock \emph{Developmental Cell}. 2017;40:429--438.
\newblock doi:10.1016/j.devcel.2016.12.012.

\bibitem{Ben_Naim1992}
Ben-Naim E, Redner S.
\newblock Inhomogeneous two-species annihilation in the steady state.
\newblock \emph{Journal of Physics A}. 1992;25:L575--L583.
\newblock doi:10.1088/0305-4470/25/9/012.

\bibitem{Gierer1972}
Gierer A, Meinhardt H.
\newblock A theory of biological pattern formation.
\newblock \emph{Kybernetik}. 1972;12:30--39.
\newblock doi:10.1007/BF00289234.

\bibitem{Meinhardt2012}
Meinhardt H.
\newblock Turing{\textquotesingle}s theory of morphogenesis of 1952 and the
  subsequent discovery of the crucial role of local self-enhancement and
  long-range inhibition.
\newblock \emph{Interface Focus}. 2012;2:407--416.
\newblock doi:10.1098/rsfs.2011.0097.

\bibitem{Landge2020}
Landge AN, Jordan BM, Diego X, M{\"{u}}ller P.
\newblock Pattern formation mechanisms of self-organizing reaction-diffusion
  systems.
\newblock \emph{Developmental Biology}. 2020;460:2--11.
\newblock doi:10.1016/j.ydbio.2019.10.031.

\bibitem{Nesterenko2017}
Nesterenko AM, Kuznetsov MB, Korotkova DD, Zaraisky AG.
\newblock Morphogene adsorption as a {T}uring instability regulator:
  theoretical analysis and possible applications in multicellular embryonic
  systems.
\newblock \emph{{PLOS} {ONE}}. 2017;12:e0171212.
\newblock doi:10.1371/journal.pone.0171212.

\end{thebibliography}
\end{document}